\documentclass{article}
\usepackage[utf8]{inputenc}
\usepackage[T1]{fontenc}
\usepackage{microtype}
\usepackage{graphicx}
\usepackage{subfigure}
\usepackage{booktabs} 
\usepackage{mathtools}
\usepackage[colorinlistoftodos]{todonotes}
\usepackage{fullpage}
\usepackage{times}
\usepackage{helvet}
\usepackage{courier}
\usepackage[hyphens]{url}
\usepackage{xcolor}
\usepackage{color, colortbl}

\usepackage{bbm}

\usepackage[usenames,dvipsnames]{pstricks}
\usepackage{epsfig}
\usepackage{pst-grad} 
\usepackage{pst-plot} 
\usepackage[space]{grffile} 
\usepackage{etoolbox} 
\makeatletter 
\patchcmd\Gread@eps{\@inputcheck#1 }{\@inputcheck"#1"\relax}{}{}
\makeatother

\usepackage{caption}

\usepackage{amsfonts,amsthm,amssymb}
\usepackage{amsmath}
\DeclareMathOperator*{\argmax}{arg\,max}

\usepackage{nicefrac}
\usepackage{algorithm}
\usepackage[noend]{algpseudocode}
\usepackage{enumerate}
\usepackage{CJK}
\newcommand{\norm}[1]{\left\lVert#1\right\rVert}
\usepackage{mdframed}

\usepackage{multirow}
\usepackage{longtable}

\newtheorem{lemma}{Lemma}[section]
\newtheorem{theorem}[lemma]{Theorem}

\newtheorem{definition}[lemma]{Definition}

\newtheorem{observation}{Observation}

\allowdisplaybreaks

\newenvironment{proofof}[1]{\smallskip\noindent{\bf Proof of #1}}%
        {\hspace*{\fill}$\Box$\par}

\newcommand{\cC}{{C}}
\newcommand{\cW}{{W}}

\newcommand{\cB}{{\mathcal B}}
\newcommand{\cG}{{\mathcal G}}

\newcommand{\pr}{{\mathbf{Pr}}}
\newcommand{\kp}{{k'}}

\newcommand{\opt}{\mathrm{OPT}}

\newcommand{\col}{\mathrm{col}}

\newcommand{\fhtw}{\mathrm{fhtw}}
\newcommand{\eps}{\epsilon}
\newcommand{\weight}{\omega}

\newcommand{\ratio}{f}

\title{Relational  Algorithms for k-means Clustering}
\author{Benjamin Moseley \\ Carnegie Mellon University\\
moseleyb@andrew.cmu.edu\and Kirk Pruhs~\thanks{Supported in part by NSF grants CCF-1421508 and CCF-1535755, and an IBM Faculty Award.}\\ University of Pittsburgh\\kirk@cs.pitt.edu \and Alireza Samadian \\University of Pittsburgh \\ samadian@cs.pitt.edu \and Yuyan Wang \\ Carnegie Mellon University \\ yuyanw@andrew.cmu.edu}
\date{May 2021}

\usepackage{graphicx}

\begin{document}
\maketitle
\begin{abstract}



This paper gives a $k$-means approximation algorithm that is efficient in the \emph{relational algorithms model}. This is an algorithm that operates directly on a relational database without performing a join to convert it to a matrix whose rows represent the data points.  The running time is potentially exponentially smaller than $N$, the number of data points to be clustered that the relational database represents.

Few relational algorithms are known and this paper offers techniques for designing relational algorithms as well as characterizing their limitations.  We show that given two data points as cluster centers, if we cluster points according to their closest centers, it is NP-Hard to approximate the number of points in the clusters on a general relational input. This is trivial for conventional data inputs and this result exemplifies that standard algorithmic techniques may not be directly applied when designing an efficient relational algorithm. This paper then introduces a new method that leverages rejection sampling and the $k$-means++ algorithm to construct a $O(1)$-approximate $k$-means solution.


\end{abstract}

\section{Introduction}
\label{section:intro}

Kaggle  surveys~\cite{KaggleSurvey} show that 
the majority of learning tasks faced by data scientists involve \emph{relational data}. Conventional formats usually represent data with multi-dimensional points where each dimension corresponds to a feature of the data. In contrast, a \textbf{relational database} consists of tables $T_1, T_2, \ldots, T_m$ where the features could be stored partially in the tables. The columns in each table are a subset of features\footnote{In relational database context the columns are also referred to as \emph{attributes} but here we call them features per the tradition of broader communities.} and the rows are data records for these features. The underlying data is represented by the \textbf{design matrix} $J=T_1 \Join \dots \Join T_m$ where each row in $J$ can be interpreted as a data point. Here the \textbf{join} ($\Join$) is a binary operator on two tables $T_i$ and $T_j$. The result of the join is the set of all possible concatenations of two rows from $T_i$ and $T_j$ such that they are equal in their common columns/features. If $T_i$ and $T_j$ have no common columns their join is the cross product of all rows. See Table \ref{table:join} for an example of join operation on two tables.

Almost all learning tasks are designed for data in matrix format.  The current standard practice for a data scientist is the following.
\begin{table}[H]
\begin{mdframed}
\noindent
\textbf{Standard Practice:}
\begin{enumerate}
    \item 
Extract the data points from the relational database by taking the join of all tables to find the design matrix $J=T_1 \Join \dots \Join T_m$. 
\item
Then interpret each row of $J$ as a point in a Euclidean space and the columns as the dimensions, corresponding to the features of data.  
\item
Import this design matrix $J$ into a standard  algorithm.
\end{enumerate}
\end{mdframed}
\end{table}


\begin{table}[t]
\centering
\begin{tabular}{|c|c|}
\hline
\rowcolor[HTML]{FFFFC7} 
\multicolumn{2}{|c|}{\cellcolor[HTML]{FFFFC7}$T_1$} \\ \hline
\rowcolor[HTML]{FFFFC7} 
$f_1$                    & $f_2$                    \\ \hline
1                        & 1                        \\ \hline
2                        & 1                        \\ \hline
3                        & 2                        \\ \hline
4                        & 3                        \\ \hline
5                        & 4                        \\ \hline
\end{tabular}
\quad
\begin{tabular}{|c|c|}
\hline
\rowcolor[HTML]{FFFFC7} 
\multicolumn{2}{|c|}{\cellcolor[HTML]{FFFFC7}$T_2$} \\ \hline
\rowcolor[HTML]{FFFFC7} 
$f_2$                    & $f_3$                    \\ \hline
1                        & 1                        \\ \hline
1                        & 2                        \\ \hline
2                        & 3                        \\ \hline
5                        & 4                        \\ \hline
5                        & 5                        \\ \hline
\end{tabular}
\quad
\begin{tabular}{|c|c|c|}
\hline
\rowcolor[HTML]{FFFFC7} 
\multicolumn{3}{|c|}{\cellcolor[HTML]{FFFFC7}$T_1 \Join T_2$} \\ \hline
\rowcolor[HTML]{FFFFC7} 
$f_1$                    & $f_2$   & $f_3$                  \\ \hline
1                        & 1   & 1                      \\ \hline
1                        & 1   & 2                      \\ \hline
2                        & 1   & 1                      \\ \hline
2                        & 1    & 2                     \\ \hline
3                        & 2     & 3                    \\ \hline
\end{tabular}
\caption{A join of tables $T_1$ and $T_2$. Each has $5$ rows and $2$ features, sharing $f_2$.  The join has all features from both tables. The rows with $f_2=x$ in the join is the cross product of all rows with $f_2=x$ from $T_1$ and $T_2$. For example, for $f_2 = 1$, the four rows in $T_1 \Join T_2$ has $(f_1, f_3)$ values $\{(1, 1), (1, 2), (2, 1), (2, 2)\}$, this is the cross product of $f_1 \in \{1, 2\}$ from $T_1$ and $f_3 \in \{1, 2\}$ from $T_2$.}
\label{table:join} \vspace{-.5cm}
\end{table}

A relational database is a highly compact data representation format. The size of $J$ can be exponentially larger than the input size of the relational database \cite{atserias2008size}. Extracting $J$ makes the standard practice inefficient.
Theoretically, there is a potential for exponential speed-up by running algorithms \emph{directly} on the tables in relational data. We call such algorithms \textbf{relational algorithms} if their running time is polynomial in  the size of tables  when the database is \emph{acyclic}. Acyclic databases will be defined shortly. This leads to the following exciting algorithmic question.



\begin{mdframed}
\noindent
\textbf{The Relational Algorithm Question:}
\renewcommand{\labelenumi}{\Alph{enumi}.}
\begin{enumerate}
    \item 
Which standard  algorithms can  be implemented as relational algorithms? 
\item
For standard algorithms that
are \emph{not} implementable by relational algorithms, is there an alternative efficient
relational algorithm that has similar performance?
\end{enumerate}
\end{mdframed}

This question has recently been of interest to the community.  However, few algorithmic techniques are known.  Moreover, we do not have a good understanding of which problems can be solved on relational data and which cannot.  Relational algorithm design has a interesting combinatorial structure that requires a deeper understanding.

We design a relational algorithm for $k$-means. It has a polynomial time complexity for \textbf{acyclic} relational databases. The relational database is acyclic if there exists a tree with the following properties. There is exactly one node in the tree for each table.  Moreover, for any feature (i.e. column) $f$, let $V(f)$ be the set of nodes whose corresponding tables contain feature $f$. The subgraph induced on $V(f)$ must be a connected component. Acyclicity can be easily checked, as the tree can be found in polynomial time if it exists \cite{yu1979algorithm}.


Luckily, most of the natural database schema are acyclic or nearly acyclic.   Answering seemingly simple questions on general (cyclic) databases, such as if the join is empty or not is NP-Hard.  For general databases, efficiency is measured in terms of the \textbf{fractional hypertree width} of the database (denoted by ``fhtw'')\footnote{See Appendix~\ref{sect:dbbackground} for a formal definition.}.  This measures how close the database structure is to being acyclic. This parameter is $1$ for acyclic databases and larger as the database is farther from being acyclic.   

State-of-the-art algorithms for queries as simple as counting the number of rows in the design matrix have linear dependency on $n^{\text{fhtw}}$ where $n$ is the \emph{maximum} number of rows in all input tables \cite{FAQ}.  Running in time linear in  $n^{\text{fhtw}}$ is the goal, as fundamental barriers need to be broken to be faster.  Notice that this is polynomial time when fhtw is a fixed constant (i.e. nearly acyclic).  Our algorithm has linear dependency on $n^{\text{fhtw}}$, matching the state-of-the-art.

\medskip
\noindent \textbf{Relational Algorithm for $k$-means:} 
$k$-means is perhaps the most widely used data mining algorithm  (e.g.  $k$-means is one of the few models in Google's BigQuery ML package~\cite{bigqueryml}). The input to the $k$-means problem consists of a
collection $S$ of points in a Euclidean space and a positive integer $k$. A feasible output is $k$  points
$c_1, \ldots, c_k$, which we call \textbf{centers}.   The objective is to choose the centers to minimize the aggregate squared
distance from each original point to its nearest center. 

Recall extracting all data points could take time exponential in the size of a relational database. Thus, the problem is to find the cluster centers without fully realizing all data points the relational data represents.

\textit{ }\cite{Rkmeans} was the first paper to give a non-trivial $k$-means algorithm that works on relational inputs. The paper gives an $O(1)$-approximation.   The algorithm's running time has a superlinear dependency on $k^d$ when the tables are acyclic and thus is not polynomial.  Here $k$ is the number of cluster centers and $d$ is the dimension (a.k.a number of features) of the points.  This is equivalently the number of distinct columns in the relational database.  For a small number of dimensions, this algorithm is a large improvement over the standard practice and  they showed the algorithm gives up to 350x speed up on real data versus performing the query to extract the data points (not even including the time to cluster the output points).  

Several questions remain.  Is there a relational algorithm for $k$-means?  What algorithmic techniques can we use as building blocks to design relational algorithms?  Moreover, how can we show some problems are hard to solve using a relational algorithm?

\medskip \noindent \textbf{Overview of Results:} 
The main result of the paper is the following.

\begin{theorem} \label{thm:main_thm1}
Given an acyclic relational database with tables $T_1, T_2, \ldots T_m$ where the design matrix $J$ has $N$ rows and $d$ columns.  Let $n$ be the maximum number of rows in any table. Then there is a randomized algorithm running in time polynomial in $d$, $n$ and $k$ that computes an $O(1)$ approximate $k$-means clustering solution with high probability.
\end{theorem}

In appendix  \ref{sect:dbbackground},  we discuss the algorithm's time complexity for cyclic databases. To illustrate the challenges for finding such an algorithm as described in the prior theorem, even when the database is acyclic, consider the following theorem.

\begin{theorem}\label{thm:hardcount}
Given an acyclic relational database with tables $T_1, T_2, \ldots T_m$ where the design matrix $J$ has $N$ rows and $d$ columns.  Given $k$ centers $c_1,\dots,c_k$, let $J_i$ be the set of points in $J$ that are closest to $c_i$ for $i \in [k]$.  It is $\#P$-Hard to compute $|J_i|$ for $k\geq 2$ and $NP$-Hard to approximate $|J_i|$ to any factor for $k \geq 3$.
\end{theorem}

 
You may find the proof in Section \ref{subsection:hardness_computing_weights}. We show it by reducing a $NP$-Hard problem to the problem of determining if $J_i$ is empty or not.   Counting the points closest to a center is a fundamental building block in almost all $k$-means algorithms. Moreover, we show even performing one iteration of the classic Lloyd's algorithm is $\#P$-Hard in Appendix~\ref{sect:Lloyds}.   Together this necessitates the design of new techniques to address the main theorem, shows that  seemingly trivial algorithms are difficult  relationally,  and suggests computing a coreset is the right approach for the problem as it is difficult to cluster the data directly.

 \medskip
\noindent \textbf{Overview of Techniques:} We first compute a \textbf{coreset}  of all points in $J$.  That is, a collection of points with weights such that if we run an $O(1)$ approximation algorithm on this weighted set, we will get a $O(1)$ approximate solution for all of $J$.  To do so, we sample points according to the principle in $k$-means++ algorithm and assign weights to the points sampled.  The number of points chosen will be $\Theta(k\log N)$. Any $O(1)$-approximate weighted $k$-means algorithm can be used on the  coreset to give Theorem \ref{thm:main_thm1}. 

\medskip
\noindent \textbf{k-means++:}  $k$-means++ is a well-known $k$-means algorithm \cite{DBLP:conf/soda/ArthurV07,aggarwal2009adaptive}.   The algorithm iteratively chooses centers $c_1, c_2, \ldots$. The first center  $c_1$ is picked uniformly from $J$. Given that $c_1, \ldots, c_{i-1}$ are picked, a point $x$ is picked 
as $c_i$ with probability
$P(x) = \frac{L(x)}{Y}$
where $L(x) = \min_{j \in [i-1]} ( \norm{x-c_j}_2^2)$ and
$Y = \sum_{x \in J} L(x)$. Here $[i-1]$ denotes $\{1,2, \ldots, i-1\}$.  

Say we sample $\Theta(k \log N)$ centers according to this distribution, which we call the \textbf{$k$-means++ distribution}.  It was shown in  \cite{aggarwal2009adaptive}  that if we cluster the points by assigning them to their closest centers, the total squared distance between points and their cluster centers is at most $O(1)$ times the optimal $k$-means cost with high probability. Note that this is not a feasible $k$-means solution because more than $k$ centers are used. However, leveraging this, the work showed that we can construct a coreset by weighting these centers according to the number of points in their corresponding clusters.  

We seek to mimic this approach with a relational algorithm. Let's focus on one iteration where we want to sample the center $c_i$ given $c_1, \ldots, c_{i-1}$ according to the $k$-means++ distribution. Consider the assignment of every point to its closest center in $c_1, \ldots, c_{i-1}$. Notice that the $k$-means++ probability is determined by this assignment. Indeed, the probability of a point being sampled is the cost of assigning this point to its closest center ($\min_{j \in [i-1]} \norm{x-c_j}_2^2$) normalized by $Y$. $Y$ is the summation of this cost over all points.  

The relational format makes this distribution difficult to compute without the design matrix $J$. It is hard to efficiently characterize which points are closest to which centers. The assignment \emph{partitions} the data points according to their closest centers, where each partition may not be easily represented by a compact relational database (unlike $J$).

\medskip
\noindent \textbf{A Relational k-means++ Implementation:} Our approach will sample every point according to the $k$-means++ distribution without computing this distribution directly. Instead, we use \textbf{rejection sampling} \cite{casella2004generalized}, which allows one to sample from a ``hard'' distribution $P$ using an ``easy'' distribution $Q$.  Rejection sampling works by sampling from $Q$ first, then reject the sample with another probability used to bridge the gap between $Q$ and $P$. The process is repeated until a sample is accepted. In our setting, $P$ is the $k$-means++ distribution, and we need to find a $Q$ which could be sampled from efficiently with a relational algorithm (without computing $J$). Rejection sampling theory shows that for the sampling to be efficient, $Q$ should be close to $P$ point-wise to avoid high rejection frequency. In the end, we will \emph{perfectly simulate} the $k$-means++ algorithm.
     
We now describe the intuition for designing such a $Q$. Recall that $P$ is determined by the assignment of points to their closest centers. We will approximate this assignment up to a factor of $O(i^{2}d)$ when sampling the $i^{th}$ center $c_i$, where $d$ is the number of columns in $J$. Intuitively, the approximate assignment makes things easier since for any center we can easily find the points assigned to it using an efficient relational algorithm. Then $Q$ is found by normalizing the squared distance between each point and its assigned center. 

The approximate assignment is designed as follows. Consider the $d$-dimensional Euclidean space where the data points in $J$ are located. The algorithm divides space into a \textbf{laminar} collection of \textbf{hyper-rectangles}\footnote{A laminar set of hyper-rectangles means any two hyper-rectangles from the set either have no intersection, or one of them contains the other.} (i.e., $\{x \in \mathcal{R}^d: v_j \leq x_j \leq w_j, j=1, \ldots, d\}$, here $x_j$ is the value for feature $f_j$). We assign each hyper-rectangle to a center.  A point assigns itself to the center that corresponds to the \emph{smallest} hyper-rectangle containing the point. 

The key property of hyper-rectangles that benefits our relational algorithm is: we can efficiently represent all points from $J$ inside any hyper-rectangle by removing some entries in each table from the original database and taking the join of all tables. For example, if a hyper-rectangle has constraint $v_j \leq x_j \leq w_j$, we just remove all the rows with value outside of range $[v_j, w_j]$ for column $f_j$ from the tables containing column $f_j$. The set of points assigned to a given center can be found by adding and subtracting a laminar set of hyper-rectangles, where each hyper-rectangle can be represented by a relational database.

\medskip
\noindent \textbf{Weighting the Centers:} We have sampled a good set of cluster centers. To get a coreset, we need to assign weights to them. As we have already mentioned, assuming $P \ne \#P$, the weights cannot be computed relationally. In fact, they cannot be approximated up to any factor in polynomial time unless $P = NP$. Rather, we design an alternative relational algorithm for computing the weights.  Each weight will not be an approximate individually, but we prove that the weighted centers form an $O(1)$-approximate coreset in aggregate.

The main algorithmic idea is that for each center $c_i$ we generate a collection of hyperspheres around $c_i$ containing geometrically increasing numbers of points. The space is then partitioned using these hyperspheres where each partition contains a portion of points in $J$. Using the algorithm from   \cite{abokhamis2020approximate}, we then sample a poly-log sized collection of points from each partition, and use this subsample to estimate the fraction of the points in this partition which are closer to $c_i$ than any other center. The estimated weight of $c_i$ is aggregated accordingly. 

\medskip
\noindent \textbf{Paper Organization:} As relational algorithms are relatively new, we begin with some special cases which help the reader build intuition.  In Section~\ref{subsec:warmup}
 we give a warm-up by showing how to implement $1$-means++ and $2$-means++ (i.e. initialization steps of $k$-means++). In this section, we also prove Theorem~\ref{thm:hardcount} as an example of the limits of relational algorithms. In Section~\ref{section:intro:background} we go over background on relational algorithms that our overall algorithm will leverage.  In Section~\ref{sec:kmeans++} we give the $k$-means++ algorithm via rejection sampling.  Section~\ref{sec:algoverview} shows an algorithm to construct the weights and then analyze this algorithm.  Many of the technical proofs appear in the appendix due to space.


\section{Warm-up: Efficiently Implementing 1-means++ and 2-means++}
\label{subsec:warmup}

This section is a warm-up to understand the combinatorial structure of relational data. We will show how to do $k$-means++ for $k\in \{1,2\}$ (referred to as 1- and 2-means++) on a simple join structure. We will also show the proof of Theorem~\ref{thm:hardcount} which states that counting the number of points in a cluster is a hard problem on relational data. 

First, let us consider relationally implementing
1-means++ and 2-means++.  For better illustration, we consider a special type of acyclic table structure named \textbf{path join}. The relational algorithm used will be generalized to work on more general join structures when we move to the full algorithm in Section \ref{sec:kmeans++}.

In a path join each table $T_i$ has two features/columns $f_i$, and $f_{i+1}$. Table $T_i$ and $T_{i+1}$ then share a common column $f_{i+1}$. Assume for simplicity that each table $T_i$ contains $n$ rows.
The design matrix  $J=T_1 \Join T_2 \Join \ldots \Join T_m$ 
has $d= m+1$ features, one for each feature (i.e. column) in the tables. 

Even with this simple structure, the size of the design matrix $J$ could still be exponential in the size of database - $J$ could contain up to $n^{m/2}$ rows , and $d n^{m/2}$ entries. Thus the standard practice could require time and space $\Omega(mn^{m/2})$ in the worst case.

\begin{table}[h]
\label{table:intro}
\centering
\begin{tabular}{|c|c|}
\hline
\rowcolor[HTML]{FFFFC7} 
\multicolumn{2}{|c|}{\cellcolor[HTML]{FFFFC7}$T_1$} \\ \hline
\rowcolor[HTML]{FFFFC7} 
$f_1$                    & $f_2$                    \\ \hline
1                        & 1                        \\ \hline
2                        & 1                        \\ \hline
3                        & 2                        \\ \hline
4                        & 3                        \\ \hline
5                        & 4                        \\ \hline
\end{tabular}
\quad
\begin{tabular}{|c|c|}
\hline
\rowcolor[HTML]{FFFFC7} 
\multicolumn{2}{|c|}{\cellcolor[HTML]{FFFFC7}$T_2$} \\ \hline
\rowcolor[HTML]{FFFFC7} 
$f_2$                    & $f_3$                    \\ \hline
1                        & 1                        \\ \hline
1                        & 2                        \\ \hline
2                        & 3                        \\ \hline
5                        & 4                        \\ \hline
5                        & 5                        \\ \hline
\end{tabular}
\quad
\begin{tabular}{|c|c|c|}
\hline
\rowcolor[HTML]{FFFFC7} 
\multicolumn{3}{|c|}{\cellcolor[HTML]{FFFFC7}$J=T_1 \Join T_2$} \\ \hline
\rowcolor[HTML]{FFFFC7} 
$f_1$                    & $f_2$   & $f_3$                  \\ \hline
1                        & 1   & 1                      \\ \hline
1                        & 1   & 2                      \\ \hline
2                        & 1   & 1                      \\ \hline
2                        & 1    & 2                     \\ \hline
3                        & 2     & 3                    \\ \hline
\end{tabular}
\quad
\raisebox{-0.5\height}{\includegraphics[scale=0.2]{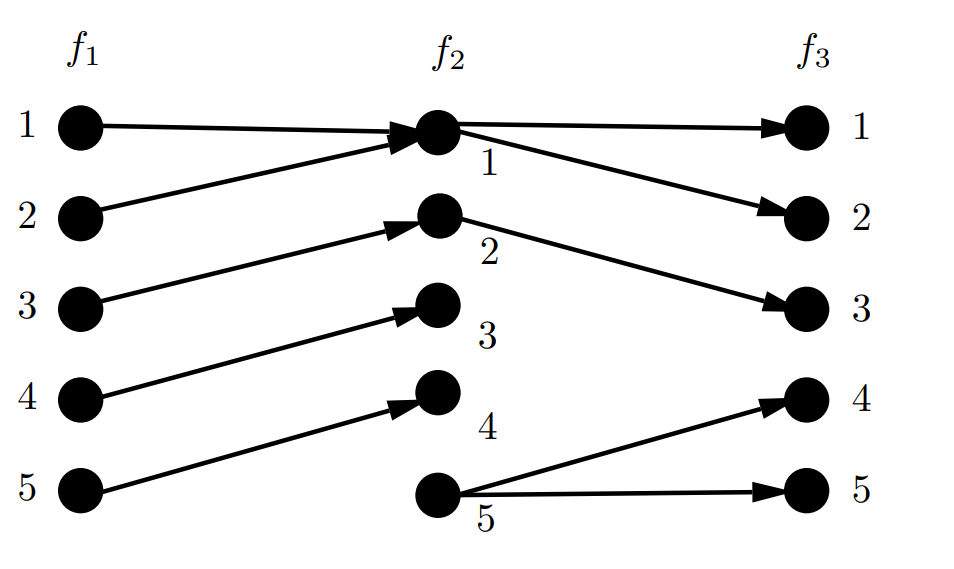}}
\caption{A path join instance where the two tables $T_1$ and $T_2$ have $m=2$ and $n=5$. This shows $T_1$, $T_2$, the design matrix $J$, and the resulting layered directed graph $G$.  \emph{Every} path from the left most layer to the right most layer of this graph $G$  corresponds to one data point for the clustering problem (i.e. a row of the design matrix).} \vspace{-.5cm}
\end{table}

\medskip
\noindent \textbf{Graph Illustration of the Design Matrix:} Conceptually consider a directed acyclic graph $G$, where there is one layer of nodes corresponding to each feature $f_i(i=1, \ldots, d)$, and edges only point from nodes in layer $f_i$ to layer $f_{i+1}$. 

The nodes in $G$ correspond to feature values, and edges in $G$ correspond to rows in tables. There is one vertex $v$ in layer $f_i$ for each value that appears in column $f_i$ in table $T_{i-1}$ or $T_i$, and one edge pointing from $u$ in layer $f_i$ to $v$ in layer $f_{i+1}$, if $(u,v)$ is a row in table $T_i$. Then, there is a one-to-one correspondence between \textbf{full paths} in $G$ (paths from layer $f_1$ to layer $f_d$) and rows in the design matrix.

\medskip
\noindent \textbf{A Relational Implementation of 1-means++:}
Implementing the 1-means++ algorithm is equivalent to \emph{generating a full path uniformly at random from $G$}. We generate this path by iteratively picking a row from table $T_1, \ldots, T_m$, corresponding to picking an arc pointing from layer $f_1$ to $f_2$, $f_2$ to $f_3$, ..., such that concatenating all picked rows (arcs) will give a point in $J$ (full path in $G$).

To sample a row from $T_1$, for every row $r \in T_1$, consider $r \Join J$, which is all rows in $J$ whose values in columns $(f_1, f_2)$ are equivalent to $r$. Let the function $F_1(r)$ denote the total number of rows in $r \Join J$. This is also the number of full paths passing arc $r$. Then, every $r \in T_1$ is sampled with probability $\frac{F_1(r)}{\sum_{r' \in T_1}F_1(r')}$, notice $\sum_{r' \in T_1}F_1(r')$ is the total number of full paths. Let the picked row be $r_1$.

After sampling $r_1$, we can conceptually throw away all other rows in $T_1$ and focus only on the rows in $J$ that uses $r_1$ to concatenate with rows from other tables (i.e., $r_1 \Join J$). For any row $r \in T_2$, let the function $F_2(r)$ denote the number of rows in $r \Join r_1 \Join J$, also equivalent to the total number of full paths passing arc $r_1$ and $r$. We sample every $r$ with probability $\frac{F_2(r)}{\sum_{r' \in T_2}F_2(r')}$. Notice that $\sum_{r' \in T_2}F_2(r')=F_1(r_1)$, the number of full paths passing arc $r_1$. Repeat this procedure until we have sampled a row in the last table $T_m$: for table $T_i$ and $r \in T_i$, assuming we have sampled $r_1, \ldots, r_{i-1}$ from $T_1, \ldots, T_{i-1}$ respectively, throw away all the other rows in previous tables and focus on $r_1\Join \ldots \Join r_{i-1} \Join J$. $F_i(r)$ is the number of rows in $r \Join r_1\Join \ldots \Join r_{i-1} \Join J$ and $r$ is sampled with probability proportional to $F_i(r)$. It is easy to verify that every full path is sampled uniformly.

For every table $T_i$ we need to find the function $F_i(\cdot)$ which is defined on all its rows. There are $m$ such functions. For each $F_i(\cdot)$, we can find all $F_i(r)$ values for $r \in T_i$ using a one-pass dynamic programming and then sample according to the values. Repeating this procedure $m$ rounds completes the sampling process. This gives a polynomial time algorithm. 


\medskip
\noindent \textbf{A Relational Implementation for 2-means++:}
Assume $x=(x_1, \ldots, x_d)$ is the first center sampled and now we want to sample the second center. By $k$-means++ principles, any row $r \in J$ is sampled with probability $\frac{\|r - x\|^2}{\sum_{r' \in J}\|r' - x\|^2}$. For a full path in $G$ corresponding to a row $r \in J$ we refer to $\|r - x\|^2$ as the \textbf{aggregated cost} over all $d$ nodes/features.

Similar to $1$-means++, we pick one row in each table from $T_1$ to $T_m$ and putting all the rows together gives us the sampled point. Assume we have sampled the rows $r_1, r_2, \ldots, r_{i-1}$ from the first $i-1$ tables and we focus on all full paths passing $r_1, \ldots, r_{i-1}$ (i.e., the new design matrix $r_1 \Join \ldots \Join r_{i-1} \Join J$). In $1$-means++, we compute $F_i(r)$ which is the total number of full paths passing arc $r_1, \ldots, r_{i-1}, r$ (i.e., $r \Join r_1 \Join \ldots \Join r_{i-1} \Join J$.) and sample $r \in T_i$ from a distribution normalized using $F_i(r)$ values. In $2$-means++, we define $F_i(r)$ to be the summation of aggregated costs over all full paths which pass arcs $r_1, \ldots, r_{i-1}, r$.  We sample $r \in T_i$ from a distribution normalized using $F_i(r)$ values. 

It is easy to verify the correctness. Again, each $F_i(\cdot)$ could be computed using a one-pass dynamic programming which gives the values for all rows in $T_i$ when we sample from $T_i$. This would involve $m$ rounds of such computations and give a polynomial algorithm. 


\subsection{Hardness of Relationally Computing the Weights:}
\label{subsection:hardness_computing_weights}
Here we prove Theorem~\ref{thm:hardcount}.  We will focus on showing that given a set of centers, counting the number of points in $J$ that is closest to any of them is $\#P$-hard. Due to space, see Appendix~\ref{section:ommited:proofs} for a proof of the other part of the theorem that it is hard to approximate the center weights for three centers. We prove $\#P$-Hardness by a reduction from the well known $\#P$-hard  Knapsack Counting problem. The input to the Knapsack Counting problem
consists of a set $W = \{w_1,\dots,w_h\}$ 
of nonnegative integer weights, and a nonnegative integer $L$. The output is the number of subsets of $W$ with aggregate weight at most $L$. To construct the relational instance, for each $i \in [h] $, we define the tables $T_{2i-1}$ and $T_{2i}$ as follows:

\begin{table}[h]
\label{table:nphard}
\centering
\begin{tabular}{|c|c|}
\hline
\rowcolor[HTML]{FFFFC7} 
\multicolumn{2}{|c|}{\cellcolor[HTML]{FFFFC7}$T_{2i-1}$} \\ \hline
\rowcolor[HTML]{FFFFC7} 
$f_{2i-1}$                    & $f_{2i}$                    \\ \hline
0                        & 0                        \\ \hline
0                        & $w_i$                       \\ \hline
\end{tabular}
\quad
\begin{tabular}{|c|c|}
\hline
\rowcolor[HTML]{FFFFC7} 
\multicolumn{2}{|c|}{\cellcolor[HTML]{FFFFC7}$T_{2i}$} \\ \hline
\rowcolor[HTML]{FFFFC7} 
$f_{2i}$                    & $f_{2i+1}$                    \\ \hline
0                        & 0                        \\ \hline
$w_i$                        & 0                        \\ \hline
\end{tabular}
\quad
\end{table}

Let centers $c_1$ and $c_2$ be arbitrary points such that points closer to $c_1$ than $c_2$ are
those points $p$ for which $\sum_{i=1}^d p_i \le L$. Then there are $2^h$ rows in $J$, since $w_i$ can either be selected or not selected in feature $2i$. The weight of $c_1$ is the number of points
in $J$ closer to $c_1$ than $c_2$, which is in turn exactly the number of subsets of $W$ with total weight at most $L$.

\section{Related Work and Background}

\label{section:intro:background}
\medskip \noindent \textbf{Related Work on K-means:} Constant approximations are known for the $k$-means problem in the standard computational setting~\cite{LiS16,KanungoMNPSW04}.  Although the most commonly used algorithm in practice is a local search algorithm called Lloyd's algorithm, or sometimes confusingly just called  ``the k-means algorithm''. The $k$-means++ algorithm from \cite{DBLP:conf/soda/ArthurV07} is a $\Theta(\log k)$ approximation algorithm, and is commonly used in practice to seed Lloyd's algorithm. Some coreset construction methods have been used before to design algorithms for the $k$-means problem in other restricted access computational models, including steaming \cite{GuhaMMMO03,BravermanFLSY17}, and the MPC model \cite{EneIM11,BahmaniMVKV12}, as well as speeding up sequential methods \cite{MeyersonOP04,SohlerW18}.

\medskip \noindent \textbf{Relational Algorithms for Learning Problem:}  Training different machine learning models on relational data has been studied; however, many of the proposed algorithms are not efficient under our definition of a relational algorithm. 
It has been shown that using repeated patterns in the design matrix, linear regression, and factorization machines can be implemented \cite{rendle2013scaling} more efficiently.  \cite{Kumar:2015:LGL:2723372.2723713, SystemF, khamis2018ac} has improved the relational linear regression and factorization machines for different scenarios. A unified relational algorithm for problems such as linear regression, singular value decomposition and factorization machines proposed in \cite{IndatabaseLinearRegression}. Algorithms for training support vector machine  is studied in \cite{yangtowards,linearSVM}. In \cite{cheng2019nonlinear}, a relational algorithm is introduced for Independent Gaussian Mixture Models, and they have shown experimentally that this method will be faster than materializing the design matrix.

\medskip \noindent \textbf{Relational Algorithm Building Blocks:} 
In the path join scenario, the $1$- and $2$-means++ sampling methods introduced in subsection \ref{subsec:warmup} have similar procedures: starting with the first table $T_1$, iteratively evaluate some general function $F_i(\cdot)$ defined on all rows in the table $T_i$, sample one row $r_i$ according to the distribution normalized from $F_i(\cdot)$. The function $F_i(\cdot)$ for table $T_i$ is defined on the matrix $r_1 \Join \ldots \Join r_{i-1} \Join J$ where $J$ is the design matrix. This matrix is also the design matrix of a new relational database, constructed by throwing away all rows in previous tables apart from the sampled $r_1, \ldots, r_{i-1}$.

We can generalize the computation of $F_i(\cdot)$ functions into a broader class of queries that we know could be implemented efficiently on \emph{any} acyclic relational databases, namely \textbf{SumProd queries}.  See \cite{FAQ} for more details. In the following lemmas assume the relational database has tables $T_1, \ldots, T_m$ and their design matrix is $J$, let $n$ be the maximum number of rows in each table $T_i$, $m$ be the number of tables and $d$ be the number of columns in $J$. 

\begin{definition}
For the $j^{th}$ feature ($j \in [d]$) let  $q_j: \mathbb{R} \rightarrow S$ be an efficiently computable function that maps feature values to some set $S$. Let the binary operations $\oplus$ and $\otimes$ be any operators such that   $(S,\oplus,\otimes)$ forms a commutative semiring. The value of $\bigoplus_{x\in J}\bigotimes_{j \in [d]} q_j(x_j)$ is a SumProd query. 
\end{definition}



\begin{lemma}[\cite{FAQ}] \label{lem:sumprod}
Any SumProd query can be computed efficiently in time $O(md^2 n^{\text{fhtw}} \log(n))$ where fhtw is the fractional hypertree width of the database. For acyclic databases fhtw=1 so the running time is polynomial.
\end{lemma}





Despite the cumbersome formal definition of SumProd queries, below we list their key applications used in this paper.  With a little abuse of notation, throughout this paper we use $\Psi(n, d, m)$ to denote the worst-case time bound on any SumProd queries.

\begin{lemma}
\label{lem:box_computing}
Given a point $y \in \mathcal{R}^d$ and a hyper-rectangle $b = \{x \in \mathcal{R}^d: v_i \leq x_i \leq w_i, i=1, \ldots, d\}$ where $v$ and $w$ are constant vectors, we let $J \cap b$ denote the data points represented by rows of $J$ that also fall into $b$. Pick any table $T_j$. Using one single SumProd query we can compute for all $r \in T_j$ the value $\sum_{p \in r \Join J \cap b}\norm{p - y}_2^2$. The time required is at most that required by one SumProd query,  $\Psi(n, d, m)$,
\end{lemma}

Lemma \ref{lem:box_computing} is an immediate result of Theorem \ref{thm:sumsum:query} which you may find in Appendix \ref{sect:dbbackground} and the fact that we can efficiently represent all points from $J$ inside any hyper-rectangle by removing some entries in each table from the original database and taking the join of all tables. The following lemma follows by an application of the main result in \cite{abokhamis2020approximate}.  In Appendix \ref{sect:faqai} we formally show to apply their result to give the following lemma.

\begin{lemma}[\cite{abokhamis2020approximate}]\label{lem:ball_computing}
Given a hypersphere $\{x \in \mathcal{R}^d: \|x - y_0\|^2 \leq z_0^2\}$ where $y_0$ is a given point and $z_0$ is the radius, a $(1+\epsilon)$-approximation of the number of points in $J$ that lie inside this hypersphere could be computed in $O\left( \frac{m^6  \log^4 n}{\epsilon^2} \Psi(n, d, m) \right)$ time. 
\end{lemma}

Notice that a SumProd query could be used to output either a scalar (similar to Lemma \ref{lem:ball_computing}) or a vector whose entries are function values for every row $r$ in a chosen table $T_j$ (in Lemma \ref{lem:box_computing}). We say the SumProd query is \textbf{grouped by} $T_j$ in the latter case.  

\section{The $k$-means++ Algorithm} \label{sec:kmeans++}
In this section, we describe a relational implementation of the $k$-means++ algorithm. It is sufficient to explain how center $c_i$ is picked given the previous centers $c_1, \ldots, c_{i-1}$. Recall that the $k$-means++ algorithm picks a point $x$ to be $c_i$ with probability $P(x) = \frac{L(x)}{Y}$ where $L(x) = \min_{j \in [i-1]}  \norm{x-c_j}_2^2$ and $Y = \sum_{x \in J} L(x)$  is a normalizing constant. 

The implementation consists of two parts. The first part, described in Section \ref{subsubsect:boxconstruction}, shows how to partition the $d$-dimensional Euclidean space into a laminar set of hyper-rectangles (referred to as \textbf{boxes} hereafter) that are generated around the previous centers. The second part, described in Section \ref{subsubsect:Qsampling}, samples according to the ``hard'' distribution $P$ using rejection sampling and an ``easy'' distribution $Q$. 

Conceptually, we assign every point in the design matrix $J$ to an \emph{approximately} nearest center among $c_1, \ldots, c_{i-1}$. This is done by assigning every point in $J$ to one of the centers contained in the \emph{smallest} box this point belongs to. Then $Q$ is derived using the squared distance between the points in $J$ and their assigned centers. 

For  illustration, we show the special case of when $k=3$ in Appendix \ref{section:3means}.  We refer the reader to this section as a warm-up before reading the general algorithm below.

\subsection{Box Construction}
\label{subsubsect:boxconstruction}
Here we explain the algorithm for constructing a set of laminar boxes given the centers sampled previously.  The construction is completely combinatorial. It only uses the given centers and we don't need any relational operation for the construction.

\medskip
\noindent \textbf{Algorithm Description:} Assume we want to sample the $i^{th}$ point in $k$-means++. The algorithm maintains two collections $\cG_i$ and $\cB_i$ of tuples. Each tuple consists of a box and a point in that box, called the \textbf{representative} of the box. This point is one of the previously sampled centers. One can think of the tuples in $\cG_i$ as ``active'' ones that are subject to changes and those in $\cB_i$ as ``frozen'' ones that are finalized, thus removed from $\cG_i$ and added to $\cB_i$. When the algorithm terminates, $\cG_i$ will be empty, and the boxes in $\cB_i$ will be a laminar collection of boxes that we use to define the ``easy'' probability distribution $Q$. 

The initial tuples in $\cG_i$ consist of one \emph{unit hyper-cube} (side length is $1$) centered at each previous center $c_j$, $j \in [i-1]$, with its representative point  $c_j$. Up to  scaling of initial unit hyper-cubes, we can assume that initially no pair of boxes in $\cG_i$ intersect. This property of $\cG_i$ is maintained throughout the  process. Initially $\cB_i$ is empty. Over time, the implementation keeps growing the boxes in $\cG_i$ in size and moves tuples from $\cG_i$ to $\cB_i$. 

The algorithm repeats the following steps in rounds. At the beginning of each round, there is no intersection between any two boxes in $\cG_i$. The algorithm performs a doubling step where it \textbf{doubles} every box in $\cG_i$. Doubling a box means each of its $d-1$ dimensional face is moved twice as far away from its representative. Mathematically, a box whose representative point is $y \in \mathcal{R}^d$ may be written as $\{x\in \mathcal{R}^d: y_i - v_i \leq x_i \leq y_i + w_i, i=1, \ldots, d\}$ ($v_i,w_i>0$). This box becomes $\{x\in \mathcal{R}^d: y_i - 2v_i \leq x_i \leq y_i + 2w_i, i=1, \ldots, d\}$ after doubling. 

After doubling, the algorithm performs the following operations on intersecting boxes until there are none. The algorithm iteratively picks two arbitrary intersecting boxes from $\cG_i$. Say the boxes are $b_1$ with representative $y_1$ and $b_2$ with representative $y_2$. The algorithm executes a \textbf{melding} step on $(b_1, y_1)$ and $(b_2, y_2)$, which has the following procedures:
\begin{itemize}
    \item Compute the smallest box $b_3$ in the Euclidean space that contains both $b_1$ and $b_2$. 
    \item Add $(b_3, y_1)$ to $\cG_i$ and delete $(b_1, y_1)$ and $(b_2, y_2)$ from $\cG_i$.
    \item 
    Check if $b_1$ (or $b_2$) is a box created by the doubling step at the beginning of the current round and hasn't been melded with other boxes ever since. If so, the algorithm computes a box $b_1'$ (resp. $b_2'$) from $b_1$ (resp. $b_2$) by \textbf{halving} it.  That is,  each $d-1$ dimensional face is moved so that its distance to the box's representative is halved. Mathematically, a box $\{x\in \mathcal{R}^d: y_i - v_i \leq x_i \leq y_i + w_i, i=1, \ldots, d\}$ ($v_i,w_i>0$), where vector $y$ is its representative, becomes $\{x\in \mathcal{R}^d: y_i - \frac{1}{2}v_i \leq x_i \leq y_i + \frac{1}{2}w_i, i=1, \ldots, d\}$ after halving. Then $(b_1', y_1)$ (or $(b_2', y_2)$) is added to $\cB_i$. Otherwise do nothing.
\end{itemize}
Notice that  melding decreases the size of $\cG_i$.

The algorithm terminates when there is  one tuple $(b_0, y_0)$ left in $\cG_i$, at which point the algorithm adds a box that contains the whole  space with representative $y_0$ to $\cB_i$. Note that during each round of the doubling and melding, the boxes which are added to $\cB_i$ are the ones that after doubling were melded with other boxes, and they are added at their shapes before the doubling step. 

\begin{lemma}
\label{lemma:box:intersection}
The collection of boxes in $\cB_i$ constructed by the above algorithm is laminar.
\end{lemma}

\begin{proof}

Note that right before each doubling step, the boxes in $\cG_i$ are disjoint and that is because the algorithm in the previous iteration melds all the boxes that have intersection with each other. We prove by induction that at all time, for every box $b$ in $\cB_i$ there exist a box $b'$ in $\cG_i$ such that $b \subseteq b'$. 
Since the boxes added to $\cB_i$ in each iteration are a subset of the boxes in $\cG_i$ before the doubling step and they do not intersect each other, laminarity of $\cB_i$ is a straight-forward consequence. 

Initially $\cB_i$ is empty and therefore the claim holds. Assume in some arbitrary iteration $\ell$ this claim holds right before the doubling step, then after the doubling step since every box in $\cG_i$ still covers all of the area it was covering before getting doubled, the claim holds. Furthermore, in the melding step every box $b_3$ that is resulted from melding of two boxes $b_1$ and $b_2$ covers both $b_1$ and $b_2$; therefore, $b_3$ will cover $b_1$ and $b_2$ if they are added to $\cB_i$, and if a box in $\cB_i$ was covered by either of $b_1$ or $b_2$, it will be still covered by $b_3$.
\end{proof}

The collection of boxes in $\cB_i$ can be thought of as a tree where every node corresponds to a box. The root node is the entire space. In this tree, for any box $b'$, among all boxes included by $b'$, we pick the inclusion-wise \emph{maximal} boxes and let them be the \textbf{children} of $b'$. Thus the number of boxes in $\cB_i$ is $O(i)$ since the tree has $i$ leaves, one for each center.

\subsection{Sampling}
\label{subsubsect:Qsampling}

To define our easy distribution $Q$, for any point $x \in J$, let $b(x)$ be the minimal box in $\cB_i$ that contains $x$ and $y(x)$ be the representative of $b(x)$. Define $R(x) = \norm{x-y(x)}_2^2$,
and $Q(x) = \frac{R(X)}{Z}$ where $Z = \sum_{x \in J} R(x)$ normalizes the distribution. We call $R(x)$ the \textbf{assignment cost} for $x$. We will show how to sample from target distribution $P(\cdot)$ using $Q(\cdot)$ and rejection sampling, and how to implement the this designed sampling step relationally.

\medskip 
\noindent \textbf{Rejection Sampling:}
The algorithm repeatedly samples a point $x$ with probability $Q(x)$, then either (A) rejects $x$ and resamples, or (B) accepts $x$ as the next center $c_i$ and finishes the sampling process. After sampling $x$, the probability of accepting $x$ is $\frac{ L(x)}{R(x)}$, and that of rejecting $x$ is $1-\frac{L(x)}{R(x)}$. Notice that here $\frac{L(x)}{R(x)} \leq 1$ since $R(x) = \norm{x-y(x)}_2^2 \geq \min_{j \in [i-1]}\norm{x-c_j}_2^2$.

If $S(x)$ is the the event of initially sampling $x$ from distribution $Q$, and $A(x)$ is the event of subsequently accepting $x$, 
the probability of choosing $x$ to be $c_i$ in one given round is:
\begin{align*}
    \Pr[S(x) \text{ and } A(x)] &= \Pr[A(x) \mid S(x)] \Pr[S(x)]
   = \frac{{L(x)} }{R(x)} Q(x)
  = \frac{L(x)}{Z}
\end{align*}
Thus the probability of $x$ being the accepted sample is proportional to $L(x)$, as desired.

We would like $Q(\cdot)$ to be close to $P(\cdot)$ point-wise so that the algorithm is efficient. Otherwise, the acceptance probability $\frac{L(x)}{R(x)}$ is low and it might keep rejecting samples.

\medskip 
\noindent \textbf{Relational Implementation of Sampling:}
We now explain how to relationally sample a point $x$ with probability $Q(x)$. The implementation heavily leverages Lemma \ref{lem:box_computing}, which states for given box $b^*$ with representative $y^*$, the cost of assigning all points in $r \Join J \cap b^*$ to $y^*$ for each row $r \in T_i$ can be computed in polynomial time using a SumProd query grouped by $T_i$.  Recall that we assign all points in $J$ to the representative of the smallest box they belong to. We show that the total assignment cost is computed by evaluating SumProd queries on the boxes and then adding/subtracting the query values for different boxes.

Following the intuition provided in Section \ref{subsec:warmup}, the implementation generates a single row from table $T_1, T_2, \ldots, T_m$ sequentially. The concatenation of these rows (or the join of them) gives the sampled point $x$. It is sufficient to explain assuming we have sampled $r_1, \ldots, r_{\ell-1}$ from the first $\ell-1$ tables, how to implement the generation of a row from the next table $T_\ell$. Just like $1$- and $2$-means++ in subsection \ref{subsec:warmup}, the algorithm evaluates a function $F_\ell(\cdot)$ defined on rows in $T_\ell$ using SumProd queries, and samples $r$ with probability $\frac{F_\ell(r)}{\sum_{r' \in T_\ell}F_\ell(r')}$. Again, we focus on $r_1 \Join \ldots \Join r_{\ell-1} \Join J$, denoting the points in $J$ that uses the previously sampled rows. The value of $F_\ell(r)$ is determined by points in $r \Join r_1 \Join \ldots \Join r_{\ell-1} \Join J$. 



To ensure we generate a row according to the correct distribution $Q$, we define the function $F_\ell(\cdot)$ as follows. Let $F_\ell(r)$ be the total assignment cost of all points in $r \Join r_1 \Join \ldots \Join r_{\ell - 1} \Join J$. That is, $F_\ell(r) = \sum_{x \in r \Join r_1 \Join \ldots \Join r_{\ell - 1} \Join J}R(x)$. Notice that the definition of function $F_\ell(\cdot)$ is very similar to $2$-means++ apart from that each point is no longer assigned to a given center, but the representative of the smallest box containing it.

Let $G(r, b^*, y^*)$ denote the cost of assigning all points from $r \Join r_1\Join \ldots \Join r_{\ell-1} \Join J$ that lies in box $b^*$ to a center $y^*$. By replacing the $J$ in Lemma \ref{lem:box_computing} by $r_1 \Join \ldots \Join r_{\ell-1} \Join J$, we can compute all $G(r, b^*, y^*)$ values in polynomial time using one SumProd query grouped by $T_\ell$. The value $F_\ell(r)$ can be expanded into subtraction and addition of $G(r, b^*, y^*)$ terms. The expansion is recursive. For a box $b_0$, let $H(r, b_0) = \sum_{x \in r \Join r_1\Join \ldots \Join r_{\ell-1} \Join J \cap b_0} R(x)$. Notice that $F_\ell(r) = H(r, b_0)$ if $b_0$ is the entire Euclidean space. Pick any row $r \in T_\ell$. Assume we want to compute $H(r, b_0)$ for some tuple $(b_0, y_0)\in \cB_i$. 

Recall that the set of boxes in $\cB_i$ forms a tree structure. If $b_0$ has no children this is the base case - $H(r, b_0) = G(r, b_0, y_0)$ by definition since all points in $b_0$ must be assigned to $y_0$. Otherwise, let $(b_1, y_1), \ldots, (b_q, y_q)$ be the tuples in $\cB_i$ where $b_1, \ldots, b_q$ are children of $b_0$. Notice that, by definition all points in $b_0 \setminus (\bigcup_{j \in [q]}b_j)$ is assigned to $y_0$. Then, one can check that the following equation holds for any $r$:
$$H(r, b_0) = G(r, b_0, y_0) - \sum_{j \in [q]}G(r, b_j, y_0) + \sum_{j \in [q]}H(r, b_j)$$
Starting with setting $b_0$ as the entire Euclidean space, the equation above could be used to recursively expand $H(\cdot, b_0)=F_\ell(\cdot)$ into addition and subtraction of $O(|\cB_i|)$ number of $G(\cdot, \cdot, \cdot)$ terms, where each term could be computed with one SumProd query by Lemma \ref{lem:box_computing}.



\medskip
\noindent \textbf{Runtime Analysis of the Sampling:} We now discuss the running time of the sampling algorithm simulating $k$-means++.
These lemmas show how close the probability distribution we compute is as compared to the $k$-means++ distribution.  This will help bound the running time.


\begin{lemma}
\label{lem:newbox_side_lengths}
Consider the box construction algorithm when sampling the $i^{th}$ point in the $k$-means++ simulation. Consider the end of the $j^{th}$ round where all melding is finished but the boxes have not been doubled yet. Let $b$ be an arbitrary box in $\cG_i$ and $h(b)$ be the number of centers in $b$ at this time. Let $c_a$ be an arbitrary one of these $h(b)$ centers.
Then:
\renewcommand{\labelenumi}{\Alph{enumi}.}
\begin{enumerate}
    \item 
    The distance from $c_a$ to any $d-1$ dimensional face of $b$ is at least $2^j$.
    \item
    The length of each side of $b$ is at most $ h(b) \cdot 2^{j+1}$.
\end{enumerate}
\end{lemma}

\begin{proof}

The first statement is a direct consequence of the definition of doubling and melding since at any point of time the distance of all the centers in a box is at least $2^j$. To prove the second statement, we define the assignment of the centers to the boxes as following. Consider the centers inside each box $b$ right before the doubling step. We call these centers, the centers assigned to $b$ and denote the number of them by $h'(b)$. When two boxes $b_1$ and $b_2$ are melding into box $b_3$, we assign their assigned centers to $b_3$.

We prove each side length of $b$ is at most $h'(b) 2^{j+1}$ by induction on the number $j$ of executed doubling steps. Since $h'(b) = h(b)$ right before each doubling, this will prove the second statement. The statement is obvious in the base case, $j=0$. 
The statement also obviously holds by induction
 after  a doubling step as $j$ is incremented and
 the side lengths double and the number of assigned boxes don't change. It also holds during every meld step because each
 side length of the newly created larger box is at most the
 aggregate maximum side lengths of the smaller boxes that are moved to $\cB_i$, and
 the number of assigned centers in the newly
 created larger box is
 the aggregate of the assigned centers in the two smaller boxes that are moved to $\cB_i$. Note that since for any box $b$ all the assigned centers to $b$ are inside $b$ at all times, $h'(b)$ is the number of centers inside $b$ before the next doubling.
\end{proof}

This lemma bounds the difference of the two probability distributions.
\begin{lemma}
\label{lemma:box:boundary}Consider the box generation algorithm when sampling the $i$th point in the $k$-means++ simulation.
For all points $x$, 
$R(x) \leq  O(i^2 d)\cdot L(x)$.
\end{lemma}

\begin{proof}

Consider an arbitrary point $x$. 
Let $c_\ell$, $\ell \in [i-1]$, be the center
that is closest to $x$ under the 2-norm distance. Assume $j$ is minimal 
such that just before the $(j+1)$-th doubling round, $x$ is contained in a box $b$ in $\cG_i$. 
We argue about the state of the algorithm
at two times, the time $s$ just
before doubling round $j$ and the time $t$ just
before doubling round $j+1$. 
Let $b$ be a minimal box in $\cG_i$ that contains $x$
at time $t$, and
let $y$ be the representative for box $b$. Notice that we assign $x$ to the representative of the smallest box in $\cB_i$ that contains it, so $x$ will be assigned to $y$. Indeed, none of the boxes added into $\cB_i$ before time $t$ contains $x$ by the minimality of $j$, and when box $b$ gets added into $\cB_i$ (potentially after a few more doubling rounds) it still has the same representative $y$.
By Lemma \ref{lem:newbox_side_lengths} the
squared distance from from $x$ to $r$ is at most
$(i-1)^2 d 2^{2j+2}$. So it is sufficient 
to show that the squared distance from $x$ to $c_\ell$
is $\Omega(2^j)$.

Let $b'$ be the box in
$\cG_i$ that contains $c_\ell$ at time $s$. Note that $x$ could not have been
inside $b'$ at time $s$ by the definition of $t$ and $s$.
Then by Lemma \ref{lem:newbox_side_lengths}
the distance from $c_\ell$ to the edge of $b'$
at time $t$ is at least $2^{2j-2}$, and hence the distance
from $c_\ell$ to $x$ is also at least $2^{2j-2}$
as $x$ is outside of $b'$. 
\end{proof}

The following theorem bounds the running time. 

\begin{theorem}
\label{thm:sampling_time}
The expected time complexity for running $k'$ iterations of this implementation of
$k$-means++ is $O(k'^4 dm \Psi(n, d, m))$.
\end{theorem}
\begin{proof}
When picking center $c_i$,  a point $x$ can be sampled with probability $Q(x)$ in time $O(m i \Psi(n, m, d))$. This is because the implementation samples one row from each of the $m$ tables. To sample one row we evaluate $O(|\cB_i|)$ SumProd queries, each in $O(\Psi(n,m,d))$ time. As mentioned earlier $\cB_i$ can be thought of as a tree of boxes with $i-1$ leaves, so $|\cB_i| = O(i)$.  

By Lemma \ref{lemma:box:boundary}, the probability of accepting any sampled $x$ is $\frac{L(x)}{R(x)} = \frac{1}{O(i^2 d)}$. The expected number of sampling from $Q$ until getting accepted is $O(i^2 d)$. Thus the expected time of finding $c_i$ is $O(i^3dm\Psi(n,m,d))$. Summing over $i \in [k']$, we get $O(k'^4 dm \Psi(n, m, d))$.
\end{proof}

\section{Weighting the Centers}
\label{sec:algoverview}

Our algorithm samples a collection
$\cC$ of $k'= \Theta(k\log{N})$ centers using the 
$k$-means++ sampling  described in the prior section.  We give weights to the centers to get a coreset.


Ideally,  we would compute the weights in the standard way. That is, let $w_i$ denote the number of points that are closest to point $c_i$ among all centers in $\cC$. These pairs of centers and weights $(c_i, w_i)$ are known to form a coreset. Unfortunately, as stated in Theorem \ref{thm:hardcount}, computing such $w_i$'s even approximately is $NP$ hard. Instead, we will find a different set of weights which still form a coreset and are computable.

Next we describe a relational algorithm to compute a collection $W'$ of weights, one weight $w'_i \in W'$ for each center $c_i \in \cC$. The proof that the centers with these alternative weights $(c_i, w'_i)$ also form a coreset is postponed until the appendix.


\medskip
\noindent \textbf{Algorithm for Computing Alternative Weights:}
Initialize the weight $w'_i$ for each center $c_i \in \cC$ to zero. In the $d$-dimensional Euclidean space, for each center $c_i \in \cC$, we generate a collection of hyperspheres (also named \textbf{balls}) $\{B_{i,j}\}_{j \in [\lg N]}$, where  $B_{i,j}$ contains approximately $2^j$ points from $J$. The space is then partitioned into $\{B_{i,0}, B_{i,1} - B_{i,0}, B_{i,2} - B_{i,1}, \ldots\}$. For each partition, we will sample a small number of points and use this sample to estimate the number of points in this partition that are closer to $c_i$ than any other centers, and thus aggregating $w'_i$ by adding up the numbers. Fix small constants $\epsilon, \delta > 0$. The following steps are repeated for $j \in [\lg N]$:
\begin{itemize}
\item
Let $B_{i, j}$ be a ball of radius $r_{i,j}$ centered at $c_i$. Find a $r_{i, j}$ such that the number of points in $J \cap B_{i, j}$ lies in the range $[(1-\delta)2^j, (1+\delta) 2^j]$.
This is an application of Lemma~\ref{lem:ball_computing}.
\item 
Let $\tau$ be a constant that is at least $30$. A collection $T_{i,j}$ of $ \frac{\tau}{\epsilon^2}\kp^2 \log^2 N$ ``test'' points are independently sampled following the same \textbf{approximately uniform} distribution with
replacement from every ball $B_{i,j}$. Here an ``approximately uniform'' distribution means one where every point $p$ in $B_{i,j}$ is sampled with a probability $\gamma_{p,i,j} \in [(1-\delta)/|B_{i,j}|, (1+\delta)/|B_{i,j}|]$ on each draw. This can be accomplished efficiently similar to the techniques used in Lemma~\ref{lem:ball_computing} from \cite{abokhamis2020approximate}.  Further elaboration is given in the Appendix~\ref{sect:faqai}.
\item
Among all sampled points $T_{i,j}$, find $S_{i,j}$, the set of points that lie in the \textbf{``donut''} $D_{i,j} = B_{i,j} - B_{i, j-1}$. Then the cardinality $s_{i,j} = |S_{i,j}|$ is computed. 
\item
Find $t_{i,j}$, the number of points in $S_{i,j}$ that are closer to $c_i$ than
any other center in $\cC$. 

\item
Compute the ratio  $\ratio'_{i,j} = \frac{t_{i,j}}{s_{i,j}}$ (if $s_{i, j} = t_{i,j} = 0 $ then $\ratio'_{i,j}= 0$).
\item
If $\ratio'_{i,j} \geq \frac{1}{2\kp^2 \log{N}}$ then $w'_i$ is incremented by $\ratio'_{i,j} \cdot 2^{j-1}$, else $w'_i$ stays the same.
\end{itemize}

At first glance, the algorithm appears naive: $w_i'$ can be significantly underestimated if in some donuts only a small portion of points are closest to $c_i$, making the estimation inaccurate based on sampling. However, in Section \ref{sec:weight_algo_analysis}, we prove the following theorem which shows that the alternative weights computed by our algorithm actually form a coreset.

\begin{theorem}
\label{thm:coreset_kmeans}
The centers $C$, along with the computed weights $W'$, form an $O(1)$-approximate coreset with high probability. 
\end{theorem}

The running time of a naive implementation of this algorithm would be dominated by sampling of the test points. Sampling a single test point can be accomplished with $m$ applications of the algorithm from \cite{abokhamis2020approximate} and setting the approximation error to $\delta = \epsilon/m$. Recall the running time of the algorithm from \cite{abokhamis2020approximate} is $O\left(  \frac{m^6  \log^4 n}{\delta^2} \Psi(n, d, m) \right)$. Thus, the time to sample all test points is $O\left(   \frac{k'^2 m^9  \log^6 n}{\epsilon^4} \Psi(n, d, m) \right)$. Substituting for $k'$, and noting that $N \le n^m$, we obtain a total time for a naive implementation of $O\left(   \frac{k^2 m^{11}  \log^8 n}{\epsilon^4} \Psi(n, d, m) \right)$.

\section{Analysis of the Weighting Algorithm}
\label{sec:weight_algo_analysis}

 The goal in this subsection is to prove
 Theorem \ref{thm:coreset_kmeans} which states
that the alternative weights form an $O(1)$-approximate coreset with high probability. Throughout our analysis,
``with high probability'' means that for any constant $\rho > 0$ the probability of the statement not being true can be made less than $\frac{1}{N^\rho}$ asymptotically by appropriately setting the constants in the algorithm.

Intuitively if a decent fraction of the points in 
each donut are closer to
center $c_i$ than any other center, then Theorem
\ref{thm:coreset_kmeans} can be proven
by using a  straight-forward  application of Chernoff bounds to show that each alternate weight $w'_i$ is likely close to the true weight $w_i$. 
The conceptual difficultly is if only a very small portion of points in a donut $D_{i, j}$ are closer to $c_i$ than any other points, in which case the estimated $\ratio'_{i,j} < \frac{1}{2\kp^2 \log{N}}$ and thus the ``uncounted'' points in $D_{i,j}$ would contribute no weight
to the computed weight $w'_i$. We call this the \textbf{undersampled} case. If many docuts around a center $i$ are undersampled, the computed weight $w'_i$
may well poorly approximate the actual weight $w_i$.

To address this, we need to prove that
omitting the weight from these  uncounted points does not have
a significant impact on the objective value. 
We break our proof into four parts.
The first part, described in 
subsubsection \ref{subsubsect:definingfractional},
involves conceptually defining a fractional weight $w_i^f$ for each center
$c_i \in \cC$. Each point has a weight of $1$, and instead of giving all this weight to its closes center, we allow fractionally assigning the weight to various ``near'' centers. $w_i^f$ is then the aggregated weight over all points for $c_i$.
The second part, described in subsubsection \ref{subsubsect:fractionalproperties}, establishes
various properties of the fractional weight that
we will need.
The third part, described in subsubsection  \ref{subsubsect:alternatefractional}, shows that
each fractional weight $w^f_i$ is likely
 to be closely approximated the computed weight $w'_i$. 
The fourth part, described in subsubsection \ref{subsubsect:fractionaloptimal}, shows that the fractional weights for the centers in $C$ form a $O(1)$-approximate coreset.
Subsubsection \ref{subsubsect:fractionaloptimal} also
contains the proof of
 Theorem \ref{thm:coreset_kmeans}.

\subsection{Defining the Fractional Weights} 
\label{subsubsect:definingfractional}

To define the fractional weights we first
define 
an auxiliary directed acyclic graph $G = (S, E)$
where there is one node in $S$ corresponding to each row in $J$. For the rest of this section, with a little abuse of notation we use $S$ to denote both the nodes in graph $G$, and the set of $d$-dimensional data points in the design matrix. 
 Let  $p$ be an  arbitrary point in $S - \cC$.
 Let $\alpha(p)$ denote the subscript of the center closest to $p$, i.e., if $c_i \in \cC$ is closest to $p$ then $\alpha(p)=i$. 
 Let $D_{i,j}$ be the donut around $c_{i}$ that contains
$p$.  If $D_{i,j}$ is not undersampled then $p$ will have one outgoing edge  $(p, c_{i})$.  So let us now assume that $D_{i,j}$ 
is undersampled. Defining the outgoing
edges from $p$ in this case is a bit more complicated.

Let $A_{i,j}$ be the 
points  $q \in D_{i,j}$ that are closer to 
$c_{i}$ than any other center in $C$ (i.e., $\alpha(q)=i$). 
If $j =1$ then $D_{i, 1}$ contains only
the point $p$, and the only outgoing edge
from $p$ goes to $c_i$. 
So let us now assume $j>1$. 
Let $c_h$ the center that
 is closest to the most points in 
 $D_{i, j-1}$, the next donut in toward
 $c_{i}$ from $D_{i, j}$. That is 
 $c_h = \argmax_{c_j \in \cC} \sum_{q \in D_{i, j-1}} \mathbbm{1}_{\alpha(q) = c_j }$.
 Let $M_{i, j-1}$ be points  in $D_{i, j-1}$ that are closer to $c_h$ than
 any other center. That is $M_{i, j-1}$ is the collection of $q \in D_{i, j-1}$
 such that $\alpha(q)=h$. Then there is a directed
 edge from $p$ to each point in $M_{i, j-1}$.
 Before defining how to derive the fractional weights from $G$, let us take a detour to note that $G$
 is acyclic. The proof of following lemma can be found in Appendix \ref{section:ommited:proofs}.

\begin{lemma}
\label{lemma:acyclic}
$G$ is acyclic. 
\end{lemma}
\begin{proof}
Consider a directed edge $(p,q) \in E$,
and $c_i$ be the center in $\cC$
that $p$ is closest
to, and $D_{i,j}$ the donut around $c_i$
that contains $p$. Then since $p \in D_{i, j}$
it must be the case that 
$\norm{p - c_i}_2^2 > r_{i, j-1}$.
Since $q \in B_{i,  j-1}$ it must
be the case that $\norm{q - c_i}_2^2 \le r_{i, j-1}$. Thus $\norm{p - c_i}_2^2 > \norm{q - c_i}_2^2$. Thus the closest center to
$q$ must be closer to $q$ than the closest center to $p$ is to $p$. 
Thus as one travels along a directed path
in $G$, although identify of
the closest center can change, 
the distance to the closest 
center must be monotonically decreasing.
Thus, $G$ must be acyclic. 
\end{proof}

We explain how to compute a fractional
weight
$w^f_p$ for each point $p \in S$ using the
network $G$. 
Initially each $w_p^f$ is set to 1. 
Then conceptually these weights flow
toward the sinks in $G$, splitting 
evenly over all outgoing edges at each vertex.
More formally,  the following flow step is repeated until is no longer
possible to do so: 

\medskip
\noindent
\textbf{Flow Step:}
Let $p \in S$ be an arbitrary point
that currently has positive fractional  weight
 and that has positive outdegree $h$ in $G$. 
Then for each directed edge $(p, q)$ in $G$
increment $w_q^f$ by $w_p^f / h$.  
Finally set $w_p^f$ to zero.
\medskip

As the sinks in $G$ are exactly the centers in $\cC$,  the centers in $\cC$ will be the only points 
that end up with positive fractional weight. 
Thus we use $w^f_i$ to refer to the
resulting fractional weight on center
$c_i \in \cC$.

\subsection{Properties of the Fractional Weights}
\label{subsubsect:fractionalproperties}

 Let $\ratio_{i,j}$ be the fraction of points that are closest to $c_i$ among all centers in $\cC$ in this donut $D_{i,j} = B_{i,j} - B_{i,j-1}$.
We show in Lemma~\ref{lem:kirk1} and Lemma~\ref{lem:kirk2} that with high probability, either the estimated ratio is a good approximation of $\ratio_{i,j}$, or the real ratio $\ratio_{i,j}$ is very small.

We show in Lemma~\ref{lem:max_weight} that the maximum flow through any node is bounded by $1+\eps$ when $N$ is big enough.  This follows using induction because each point has $\Omega(\kp \log{N})$ neighbors and every point can have in degree from one set of nodes per center. We further know every point that is not uncounted actually contributes to their centers weight.

\begin{lemma}\label{lem:kirk1}
With high probability either $|\ratio_{i,j}-\ratio'_{i,j}| \leq \epsilon \ratio_{i,j}$ or $\ratio'_{i,j} \le \frac{1}{2\kp^2 \log{N}}$.
\end{lemma}

To prove Lemma \ref{lem:kirk1}, we use the following Chernoff Bound.
\begin{lemma}
\label{lemma:Chernoff}
Consider Bernoulli trials $X_i, \ldots, X_n$.
Let $X = \sum_{i=1}^n X_i$
and $\mu = E[X]$.
 Then, for any $\lambda >0$:
\begin{align*}
\pr[X \geq \mu + \lambda] &\leq \exp\left(-\frac{\lambda^2}{2\mu+\lambda} \right) & \textnormal{ Upper Chernoff Bound}\\ 
\pr[X \leq \mu - \lambda] &\leq \exp\left(-\frac{\lambda^2}{3\mu}\right) & \textnormal{ Lower Chernoff Bound}
\end{align*}
\end{lemma}

\begin{proof}{Proof of Lemma \ref{lem:kirk1}:}
Fix any center $c_i \in \cC$ and $j \in [\log N]$.  By applying the low Chernoff bound 
from Lemma \ref{lemma:Chernoff} it is
straight forward to conclude that $\tau$
is large then 
 with high probability at least
a third of the test points in each 
$T_{i,j}$ are in the donut $D_{i,j}$.
That is, with high probability 
$s_{i,j} \ge  \frac{\tau}{3\epsilon^2}\kp^2 \log^2 N$ .   
So let us consider a particular $T_{i,j}$
and condition 
$s_{i,j}$ having some fixed
value that is at least $\frac{1}{3\epsilon^2}\kp^2 \log^2 N$.
So $s_{i,j}$ is conditioned on being large.

Recall $t_{i,j} = \sum_{p \in W_{i,j}}
(\mathbbm{1}_{ p \in T_{i,j} })
(\mathbbm{1}_{ \alpha(p) =i })$,
and the indicator random variables
$\mathbbm{1}_{ p \in T_{i,j} }$ 
are Bernoulli trials. 
Further note by the definition of
$\gamma_{p,i,j}$ it is the case that
$ E[t_{i,j}] = \sum_{p \in W_{i,j}} \gamma_{p,i,j}(\mathbbm{1}_{ \alpha(p) =i })$.
Further note that as the sampling of
test points is nearly uniform that
$ f_{i,j} (1-\delta) s_{i,j} \le E[t_{i,j}] \le  f_{i,j} (1+\delta) s_{i,j}$.
For notational convenience, let $\mu = E[t_{i,j}]$.
We now break the proof into three cases,
that cover the ways in which the statement
of this lemma would not be true.
For each case, we show with high probability the case does not occur. 

\textbf{Case 1: $\ratio'_{i,j} \geq \frac{1}{2\kp^2 \log{N}}$ and $\ratio_{i,j} > \frac{1-\eps}{2 k'^2 \log{N}} $ and $\ratio'_{i,j} \geq (1+\epsilon) \ratio_{i,j}$.}
We are going to prove the probability of this case happening is very low. 
If we set $\lambda = \eps \mu $, then using Chernoff bound, we have

\begin{align*}
    \pr[t_{i,j} \geq (1+\eps) \mu ] &\le \exp\left(-\frac{(\eps \mu)^2}{2\mu+\eps \mu} \right) &\textnormal{[Upper Chernoff Bound]}
    \\
    &\leq
    \exp\left(-\frac{\eps^2 (1-\delta) \ratio_{i,j} s_{i,j}}{2+\eps} \right) &\textnormal{[$ \mu \geq  (1-\delta) \ratio_{i,j} s_{i,j}$]}
    \\
    &\leq
    \exp\left(-\frac{\eps^2 (1-\delta) (1-\eps) s_{i,j}}{3(2 k'^2 \log{N})} \right) &\textnormal{[$ \ratio_{i,j} > \frac{1-\eps}{2 k'^2 \log{N}}$]}
    \\
    &\leq
    \exp\left(-\frac{\eps^2 (1-\delta) (1-\eps) \tau \kp^2 \log^2 N}{3(2 k'^2 \log{N})(3\epsilon^2)} \right) &\textnormal{[$ s_{i,j} \geq \frac{\tau}{3\epsilon^2}\kp^2 \log{N}$]}
    \\
    &=
    \exp\left(-\frac{(1-\delta) (1-\eps) \tau \log N}{18} \right) 
\end{align*}
Therefore, for $\delta \leq \eps/2 \leq 1/10$ and $\tau \geq 30$ this case cannot happen with high probability.

\textbf{Case 2: $\ratio'_{i,j} \geq \frac{1}{2\kp^2 \log{N}}$ and $\ratio_{i,j} > \frac{1-\eps}{2 k'^2 \log{N}} $ and  $\ratio'_{i,j} < (1-\epsilon) \ratio_{i,j}$.}
We can use Lower Chernoff Bound with $\lambda=\eps \mu$ to prove the probability of this event is very small. 
\begin{align*}
    \pr[t_{i,j} \leq (1-\eps) \mu ]
    &\leq \exp\left(-\frac{(\eps \mu)^2}{3\mu} \right)
    \\
    &\leq
    \exp\left(-\frac{\eps^2 (1-\delta) \ratio_{i,j} s_{i,j} }{3} \right)
    &\textnormal{[$ \mu \geq (1-\delta) \ratio_{i,j} s_{i,j}$]}
    \\
    &\leq
    \exp\left(-\frac{\eps^2 (1-\delta) (1-\eps) s_{i,j}}{3(2 k'^2 \log{N})} \right) &\textnormal{[$ \ratio_{i,j} > \frac{1-\eps}{2 k'^2 \log{N}}$]}
    \\
    &\leq
    \exp\left(-\frac{\eps^2 (1-\delta) (1-\eps) \tau \kp^2 \log^2 N}{3(2 k'^2 \log{N})(3\epsilon^2)} \right) &\textnormal{[$ s_{i,j} \geq \frac{\tau}{3\epsilon^2}\kp^2 \log{N}$]}
    \\
    &=
    \exp\left(-\frac{(1-\delta) (1-\eps) \tau \log N}{18} \right) 
\end{align*}

Therefore, for $\delta \leq \eps/2 \leq 1/10$ and $\tau \geq 30$ this case cannot happen with high probability.

\textbf{ Case 3:
\boldmath $\ratio'_{i,j} \geq \frac{1}{2\kp^2 \log{N}}$ and $\ratio_{i,j}\leq \frac{1-\epsilon}{2\kp^2 \log{N}}$\unboldmath:} 
Since $\ratio'_{i,j} = \frac{t_{i,j}}{s_{i,j}}$,  in this case:
\begin{align}
\label{eq:dax1}
    t_{i,j} \geq \frac{s_{i,}}{2\kp^2 \log{N}}
\end{align}
Since $\mu \le   f_{i,j} (1+\delta) s_{i,j}$,
in this case:
\begin{align}
\label{eq:dax2}
\mu \le   \frac{1-\epsilon}{2\kp^2 \log{N}} (1+\delta) s_{i,j}
\end{align}
Thus subtracting line \ref{eq:dax1} from line \ref{eq:dax2} we conclude that:
\begin{align}
t_{i,j} \ge \mu + \frac{(\epsilon - \delta + \epsilon \delta) s_{i,j}}{2\kp^2 \log{N}}    
\end{align}
Let $\lambda = \frac{(\epsilon - \delta + \epsilon \delta) s_{i,j}}{2\kp^2 \log{N}}$. 
We can conclude that
\begin{align*}
\pr[t_{i,j} \ge \mu + \lambda  ] &\le \exp\left(-\frac{\lambda^2}{2\mu+\lambda} \right) &\textnormal{Upper Chernoff Bound}\\
&\le \exp\left(\frac{-\lambda^2}{\frac{1-\epsilon}{2\kp^2 \log{N}} (1+\delta) s_{i,j} +\lambda} \right) &\textnormal{Using line \ref{eq:dax2}}\\
&= \exp\left(
\frac{- 
\left(
\frac{ (\epsilon - \delta + \epsilon \delta) s_{i,j}}{2 \kp^2\log{N}}
\right)^2}
{\frac{1-\epsilon}{2\kp^2 \log{N}} (1+\delta) s_{i,j} +\frac{ (\epsilon - \delta + \epsilon \delta) s_{i,j}}{2 \kp^2 \log{N}}
}\right) & \\
&= \exp\left(\frac{- \left(\frac{ (\epsilon - \delta + \epsilon \delta)^2 s_{i,j}}{ \kp^2 \log{N}}\right)}{2(1-\epsilon)(1+\delta)  + 2(\epsilon - \delta + \epsilon \delta)   } \right) \\
&\le \exp\left( \frac{ - (\epsilon - \delta + \epsilon \delta)^2 s_{i,j}}{12 \kp^2 \log{N}}\right)  &\\
&= \exp\left( \frac{ - (\epsilon - \delta + \epsilon \delta)^2 \tau\log{N} }{12 \epsilon^2 }\right) &\textnormal{Substituting our lower bound on }  s_{i,j}
\end{align*}
Therefore, for $\delta \leq \eps/2 \leq 1/10$ and $\tau \geq 30$ this case cannot happen with high probability. 

\end{proof}

The next case proves the how large $f'_{i,j}$ is when we know that $f_{i,j}$ is large.

\begin{lemma}\label{lem:kirk2}
If $\ratio_{i,j} > \frac{1+\epsilon}{2k'^2\log{N}}$ then 
with high probability $\ratio'_{i,j} \geq \frac{1}{2k'^2\log{N}}$. 
\end{lemma}
\begin{proof} We can prove that the probability of $\ratio'_{i,j} < \frac{1}{2k'^2\log{N}}$ and $\ratio_{i,j} \geq \frac{1+\epsilon}{2k'^2\log{N}}$ is small. Multiplying the conditions for this case by
$s_{i,j}$ we can conclude that 
$t_{ij} < \frac{s_{i,j}}{2\kp^2 \log{N}}$ and $\mu \geq(1-\delta) \frac{(1+\epsilon) s_{i,j}}{2\kp^2 \log{N}}$.
And thus $t_{i,j} \le \mu - \lambda$
where $\lambda = \frac{(\epsilon-\delta -\epsilon\delta) s_{i,j}}{2\kp^2 \log{N}}$. Then we can conclude
that:
\begin{align*}
\pr[t_{i,j} \leq \mu - \lambda  ] &\le \exp\left(-\frac{\lambda^2}{3\mu}\right) 
&\textnormal{[Lower Chernoff Bound]}\\
&= \exp\left(-\frac{\left(\frac{(\epsilon-\delta -\epsilon\delta) s_{i,j}}{2\kp^2 \log{N}}\right)^2}{3\mu}\right) 
& \\
&\le \exp\left(-\frac{\left(\frac{(\epsilon-\delta -\epsilon\delta) s_{i,j}}{2\kp^2 \log{N}}\right)^2}{3\frac{1-\epsilon}{2\kp^2 \log{N}} (1+\delta) s_{i,j}}\right) 
& \\
&= \exp\left(-\frac{\left(\frac{(\epsilon-\delta -\epsilon\delta)^2 s_{i,j}}{2\kp^2 \log{N}}\right)}{3(1-\epsilon)(1+\delta) }\right) 
& \\
&\le \exp\left(\frac{-(\epsilon-\delta -\epsilon\delta)^2 s_{i,j}}{12 \kp^2 \log{N}}\right) 
& \mbox{[$\delta< \eps \leq 1$]} \\
&\le \exp\left(\frac{-(\epsilon-\delta -\epsilon\delta)^2 ( \frac{\tau}{3\epsilon^2}\kp^2 \log^2 N)}{12 \kp^2 \log{N}}\right) 
& \textnormal{[Using our lower bound on $ s_{i,j}$]} \\
\end{align*}
Therefore, for $\delta \leq \eps/2 \leq 1/10$ and $\tau \geq 30$ this case cannot happen with high probability.

\end{proof}



We now seek to bound the fractional weights computed by the algorithm. Let $\Delta_i(p)$ denote the total weight received by a point $p \in S\setminus \cC$ from other nodes (including the initial weight one on $p$). Furthermore, let $\Delta_o(p)$ denote the total weight sent by $p$ to all other nodes. Notice that in the flow step $\Delta_o(p) = \Delta_i(p)$ for all $p$ in $S \setminus C$. 

\begin{lemma}
\label{property:total_weights}
Let $\Delta_i(p)$ denote the total weight received by a point $p \in S\setminus \cC$ from other nodes (including the initial weight one on $p$). Furthermore, let $\Delta_o(p)$ denote the total weight sent by $p$ to all other nodes. With high probability, for all $q\in S$, $\Delta_i(q) \leq 1 + \frac{1+2\epsilon}{\log{N}}\max_{p:(p,q) \in E}\Delta_o(p)$. 
\end{lemma}

\begin{proof}
Fix the point $q$ that redirects its weight (has outgoing arcs in $G$). Consider its direct predecessors: $P(q) = \{p: (p,q) \in E\}$.  Partition $P(q)$ as follows: $P(q) = \bigcup_{i = 1, \ldots, \kp}P_{c_i}(q)$, where $P_{c_i}(q)$ is the set of points that have flowed their weights into $q$, but $c_i$ is actually their closest center in $\cC$. Observe the following. The point $q$ can only belong to one donut around $c_i$.  Due to this, $P_{c_i}(q)$ is either empty or contains a set of points in a single donut around $c_i$ that redirect weight to $q$.

Fix $P_{c_i}(q)$ for some $c_i$. If this set is non-empty suppose this set is in the $j$-th donut around $c_i$. Conditioned on the events stated in Lemmas \ref{lem:kirk1} and \ref{lem:kirk2}, since the points in $P_{c_i}(q)$ are undersampled, we have $|P_{c_i}(q)| \leq \frac{(1+\epsilon)2^{j-1}}{2\kp^2 \log{N}}$. Consider any $p \in P_{c_i}(q)$.  Let $\beta_i$ be the number of points that $p$ charges its weight to (this is the same for all such points $p$). It is the case that $\beta_i$ is at least $\frac{(1-\delta)2^{j-1}}{2\kp}$ since $p$ flows its weights to the points that are assigned to the center that has the most number of points assigned to it from $c_i$'s $(j-1)$th donut.   

Thus, $q$ receives weight from $|P_{c_i}(q)| \leq \frac{(1+\epsilon)2^{j-1}}{2\kp^2 \log{N}}$ points and each such point gives its weight to at least $\frac{(1-\delta)2^{j-1}}{2\kp}$ points with equal split.
The total  weight that $q$ receives from points in $P_{c_i}(q)$ is at most the following. 

\begin{align*}
&\qquad \frac{2\kp}{(1-\delta)2^{j-1}}  \sum_{p \in P_{c_i}(q)} \Delta_o(p) &\\
 &\leq \frac{2\kp}{(1-\delta)2^{j-1}}  \sum_{p \in P_{c_i}(q)} \max_{p \in P_{c_i}(q) }\Delta_o(p)& \\
 &\leq \frac{2\kp}{(1-\delta)2^{j-1}} \cdot \frac{(1+\epsilon) \cdot 2^{j-1}}{2\kp^2 \log{N}} \  \max_{p \in P_{c_i}(q) }\Delta_o(p)  &\mbox{[$|P_{c_i}(q)| \leq \frac{(1+2\epsilon)2^{j-1}}{2\kp^2 \log{N}}$]}\\ 
 &\leq \frac{1+2\eps}{\kp \log N} \max_{p \in P_{c_i}(q)}\Delta_o(p)
  & \mbox{[$\delta \leq \frac{\epsilon}{2} \leq \frac{1}{10}$]}
\end{align*}

Switching the max to $\max_{p:(p,q) \in E}\Delta_o(p)$, summing over all centers $c_i \in \cC$ and adding the original unit weight on $q$ gives the lemma.



\end{proof}

The following crucial lemma  bounds the maximum weight that a point can receive. 

\begin{lemma}
\label{lem:max_weight}
Fix  $\eta$ to be a constant smaller than $\frac{\log(N)}{10}$ and $\eps <1$. Say that for all $q \in S \setminus C$ it is the case that $\Delta_o(q) =\eta \Delta_i(q)$. Then, with  high probability for any $p \in S \setminus C$ it is the case that $\Delta_i(p) \leq 1+ \frac{2\eta}{\log{N}}$.
\end{lemma}
\begin{proof}
We can easily prove this by induction on nodes.  The lemma is true for all nodes that have no incoming edges in $G$.  Now assume it is true for all nodes whose longest path that reaches them in $G$ has length $t-1$. Now we prove it for nodes whose longest path that reaches then in $G$ is $t$.  Fix such a node $q$. For any node $p$ such that $(p,q) \in E$, by induction we have $\Delta_i(p) \leq 1 +\frac{2\eta}{\log{N}}$, so $\Delta_o(p) \leq 2(1+\frac{2\eta}{\log{N}})$. By Lemma~\ref{property:total_weights}, $\Delta_i(q) \leq 1 + \frac{1+2\epsilon}{\log{N}}\max_{p:(p,q) \in E}\Delta_o(p) \leq  1 + \left (\frac{\eta(1+2\epsilon)}{\log{N}}  \right) \left (1+ \frac{2\eta}{\log{N}} \right ) = 1+\frac{\eta}{\log{N}} + \frac{\eta}{\log{N}} \cdot \frac{2(1+2\epsilon)\eta +2\epsilon}{\log{N}} \leq 1+ \frac{2\eta}{\log{N}}$.

\end{proof}

\subsection{Comparing Alternative Weights to Fractional Weights} 
\label{subsubsect:alternatefractional}

It only remains to bound the cost of mapping points to the centers they  contribute weight to.  This can be done by iteratively charging the total cost of reassigning each node with the flow.  In particular, each point will only pass its weight to nodes that are closer to their center. We can charge the flow through each node to the assignment cost of that node to its closest center, and argue that the cumulative reassignment cost bounds the real fractional assignment cost. Further, each node only has $1+\eps$ flow going through it.  This will be sufficient to bound the overall cost in Lemma~\ref{lem:bound-cost-by-alg}.

\begin{lemma}\label{lem:close-weights}
With high probability, for every center $c_i$, it is the case that the estimated weight $w'_i$ computed by the weighting algorithm is $(1 \pm 2\eps)w^f_i$ where $w^f_i$ is the fractional weight of $i$.
\end{lemma}

\begin{proof}
Apply union of bounds to Lemma \ref{lem:kirk1} and \ref{lem:kirk2}  over all $i$ and $j$. 

Fix a center $c_i$. Consider all of the points that are closest to $c_i$ and are not undersampled. Let $w^s_i$ denote the number of these points. All the incomming edges to $c_i$ in $G$, are coming from these points; therefore based on Lemma \ref{lem:max_weight}, $w^s_i \leq w^f_i \leq w^s_i (1+\frac{2}{\log(N)})$. On the other hand, $w'_i$ is $(1\pm \epsilon)$ approximation of $w^s_i$.  Therefore, $ \frac{1-\epsilon}{1+\frac{2}{\log(N)}} w^f_i\leq w'_i \leq (1+\epsilon) w^f_i$.  Assuming that $\log N$ is sufficiently larger than $\eps$, the lemma follows. 

\end{proof}

\subsection{Comparing Fractional Weights to Optimal} 
\label{subsubsect:fractionaloptimal}

Next we bound the total cost of the fractional assignment defined by the flow. According to the graph $G$, any point $p \in S$ and $c_i \in \cC$, we let $\weight(p, c_i)$ be the fraction of weights that got transferred from $p$ to $c_i$. Naturally we have $\sum_{c_i \in \cC} \weight(p, c_i) = 1$ for any $p \in S$ and the fractional weights $w_i^f = \sum_{p \in S}\weight(p, c_i)$ for any $c_i \in \cC$.

\begin{lemma} \label{lem:bound-cost-by-alg}
Let  $\phi_{opt}$ be the optimal $k$-means cost on the original set $S$. With high probability, it is the case that:
$$\sum_{p\in S}\sum_{c_i \in \cC}\weight(p, c_i)\|p - c_i\|^2 \leq 160 (1+\eps) \phi_{opt}$$
\end{lemma}

\begin{proof}
Let  $\phi^* = \sum_{p \in S} \|p - c_{\alpha(p)}\|^2$. Consider any $p \in S$ and center $c_i$ such that $\omega(p,c_i)> 0$.
Let $P$ be any path from $p$ to $c_i$ in $G$. If node $p$'s only outgoing arc is to its closest center $c_{\alpha(p)} = c_i$, then $P = p \rightarrow c_i$, we have $\sum_{c \in \cC}\weight(p, c) \|p - c\|^2 = \|p - c_{\alpha(p)}\|^2$. Otherwise assume $P = p \rightarrow q_1 \rightarrow q_2 \rightarrow \ldots \rightarrow q_\ell \rightarrow c_i$. Note that the closest center to $q_\ell$ is $c_i$.   Let $\Delta(P)$ be the fraction of the original weight of $1$ on $p$ that is given to $c_i$ along this path according to the flow of weights. As we observed in the proof of Lemma \ref{lemma:acyclic}, we have  $\|p - c_{\alpha(p)}\| > \|q_1 - c_{\alpha(p)}\| \geq \|q_1 - c_{\alpha(q_1)}\| > \|q_2 - c_{\alpha(q_1)}\| \geq \|q_2 - c_{\alpha(q_2)}\| > \ldots > \|q_\ell - c_{\alpha(q_\ell)}\|$.   This follows because for any arc $(u,v)$ in the graph, $v$ is in a donut closer to $c_{\alpha(u)}$ than the donut $u$ is in, and $v$ is closer to  $c_{\alpha(v)}$ than $c_{\alpha(u)}$. 

We make use of the relaxed triangle inequality for squared $\ell_2$ norms. For any three points $x,y,z$, we have $\|x-z\|^2 \leq 2(\|x-y\|^2 + \|y-z\|^2)$. Thus, we bound $\|p - c_i\|^2$ by
\begin{align*}
    \|p - c_i\|^2 &= \|p- c_{\alpha(p)} + c_{\alpha(p)} - q_1 + q_1 - c_i\|^2 \\
    &\leq 2\|p- c_{\alpha(p)} + c_{\alpha(p)} - q_1\|^2 + 2\|q_1 - c_i\|^2   &&\text{[relaxed triangle inequality]}\\
    & \leq 2(\|p- c_{\alpha(p)}\|+\| c_{\alpha(p)} - q_1\|)^2 + 2\|q_1 - c_i\|^2 && \mbox{[triangle inequality]}\\
    & \leq 8\|p-c_{\alpha(p)}\|^2 + 2\|q_1 - c_i\|^2  && \mbox{[$\|p- c_{\alpha(p)}\|\geq\| c_{\alpha(p)} - q_1\|$]}.
\end{align*}
    
    Applying the prior steps to each $q_i$ gives the following.
\begin{align*}
   \|p - c_i\|^2  &\leq 8(\|p - c_{\alpha(p)}\|^2 + \sum_{j=1}^\ell 2^j \|q_j - c_{\alpha(q_j)}\|^2)
\end{align*}

Let $\mathcal{P}_q(j)$ be the set of all paths $P$ that reach point $q$ using $j$ edges. If $j=0$, it means $P$ starts with point $q$. We seek to bound $\sum_{j=0}^\infty  2^j \sum_{P \in \mathcal{P}_q(j)} \Delta(P) \|q - c_{\alpha(q_j)}\|^2$. This will bound the charge on point $q$ above over all path $P$ that contains it.   

 Define  a weight function $\Delta'(p)$ for each node $p \in S \setminus C$.  This will be a new flow of weights like $\Delta$, except now the weight increases at each node. In particular,  give each node initially a weight of $1$. Let $\Delta'_o(p)$ be the total weight leaving $p$.  This will be evenly divided among the nodes that have outgoing edges from $p$.  Define $\Delta'_i(p)$ to be the weight incoming to $p$ from all other nodes plus one, the initial weight of $p$.  Set $\Delta'_o(p)$ to be $2 \Delta'_i(p)$, twice the incoming weight.  

 Lemma~\ref{lem:max_weight} implies that the maximum weight of any point $p$ is $\Delta'_i(p) \leq 1 + \frac{4}{\log N}$.  Further notice that for any $q$ it is the case that $\Delta'_i(q)= \sum_{j=0}^\infty  2^j \sum_{P \in \mathcal{P}_q(j)} \Delta(P) $. Letting $\mathcal{P}(p,c_i)$ be the set of all paths that start at $p$ to center $c_i$. Notice such paths correspond to how $p$'s unit weight goes to $c_i$.  We have $\weight(p, c_i) = \sum_{P \in \mathcal{P}(p, c_i)} \Delta(P)$. Let $\mathcal{P}$ denote the set of all paths, $\ell(P)$ denote the length of path $P$ (number of edges on $P$) , and let $P(j)$ denote the $j$th node on path $P$. Thus we have the following. 

\begin{eqnarray*}
&&\sum_{p\in S}\sum_{c_i \in \cC}\weight(p, c_i)\|p - c_i\|^2\\
&= &\sum_{p\in S}\sum_{c_i \in \cC}\sum_{P \in \mathcal{P}(p,c_i)} \Delta(P) \|p - c_i\|^2 \\
&\leq & 8\sum_{p\in S}\sum_{c_i \in \cC} \sum_{P\in \mathcal{P}(p,c_i)} \Delta(P)(\sum_{j=0}^{\ell(p) - 1}2^j \|P(j) - c_{\alpha(P(j))}\|^2 ) \\
& = & 8 \sum_{P \in \mathcal{P}} \Delta(P)(\sum_{j=0}^{\ell(p) - 1}2^j \|P(j) - c_{\alpha(P(j))}\|^2 ) \\
& = & 8\sum_{q \in S} \sum_{j=0}^{+\infty}\sum_{P \in \mathcal{P}_q(j)} 2^j \Delta(P) \|q - c_{\alpha(q)}\|^2 \\
& = & 8\sum_{q \in S} \Delta'_i(q) \|q - c_{\alpha(q)}\|^2 \\
& \leq & \sum_{q \in S} 8(1+\frac{4}{\log{N}})\|q - c_{\alpha(q)}\|^2 = 8(1+\frac{4}{\log{N}})\phi^* 
\end{eqnarray*}
Lemma \ref{lem:bound-cost-by-alg} follows because if $k' \geq 1067k\log{N}$, $\phi^* \leq 20 \phi_{opt}$ with high probability by Theorem~1 in \cite{aggarwal2009adaptive}.

\end{proof}

Finally, we prove that finding any $O(1)$-approximation solution for optimal weighted $k$-means on the set $(\cC, \cW')$ gives a constant approximation for optimal $k$-means for the original set $S$. Let $\cW^f = \{w_1^f, \ldots, w_\kp^f\}$ be the fractional weights for centers in $\cC$. Let $\phi_{\cW^f}^*$ denote the optimal weighted $k$-means cost on $(\cC, \cW^f)$, and $\phi_{\cW'}^*$ denote the optimal weighted $k$-means cost on $(\cC, \cW')$. We first prove that $\phi_{\cW^f}^* = O(1) \phi_\opt$, where $\phi_\opt$ denote the optimal $k$-means cost on set $S$.

\begin{lemma} \label{lem:const_appx_adaptive}
Let $(\cC, \cW^f)$ be the set of points sampled and the weights collected by fractional assignment $\weight$. With high probability, we have $\phi_{\cW^f}^* = O(1) \phi_\opt$.
\end{lemma}

\begin{proof}
Consider the cost of the fractional assignment we've designed. For $c_i \in \cC$, the weight is $w^f_i = \sum_{p \in S}\weight(p, c_i)$. Denote the $k$-means cost of $\weight$ by $\phi_\weight = \sum_{p \in S}\sum_{c \in \cC}\weight(p, c)\|p - c\|^2$. By Lemma \ref{lem:bound-cost-by-alg}, we have that $\phi_\weight \leq 160(1+\epsilon) \phi_\opt$. 

Intuitively, in the following we show $\phi_{\cW^f}^*$ is close to $\phi_\weight$. As always, we let $\cC_\opt$ denote the optimal centers for $k$-means on set $S$. For set of points $X$ with weights $Y:X \to \mathbb{R}^+$ and a set of centers $Z$, we let $\phi_{(X,Y)}(Z) = \sum_{x \in X}Y(x)\min_{z \in Z}\|x-z\|^2$ denote the cost of assigning the weighted points in $X$ to their closest centers in $Z$. Note that $\phi_{\cW^f}^*\leq \phi_{(\cC, \cW^f)}(\cC_\opt)$ since $\cC_\opt$ is chosen with respect to $S$. 

\begin{align*}
    \phi_{\cW^f}^* & \leq \phi_{(\cC, \cW^f)}(\cC_\opt) 
    \\
    &= \sum_{c_i \in \cC} (\sum_{p \in S}\weight(p, c_i)) \min_{c \in \cC_\opt}\|c_i - c\|^2
    & \mbox{[$w_i^f=\sum_{p \in S}\weight(p, c_i)$]}
    \\  
    &= \sum_{c_i \in \cC} \sum_{p \in S}\min_{c \in \cC_\opt}\weight(p, c_i) \|c_i - c\|^2
    \\
    & \leq \sum_{c_i \in \cC} \sum_{p \in S}\min_{c \in \cC_\opt}\weight(p, c_i) \cdot 2(\|p-c_i\|^2 + \|p - c\|^2)
    & \mbox{[relaxed triangle inequality]}\\
    & = 2\phi_\weight + 2\phi_\opt \leq 322(1+\epsilon) \phi_\opt
\end{align*}
\end{proof}

Using the mentioned lemmas, we can prove the final approximation guarantee. 
\begin{proof}[Proof of Theorem \ref{thm:coreset_kmeans}]
Using Lemma \ref{lem:close-weights}, we know $w'_i = (1 \pm 2\epsilon) w_i^f$ for any center $c_i$. Let $\cC'_k$ be $k$ centers for $(\cC, \cW')$ that is a $\gamma$-approximate for optimal weighted $k$-means. Let $\cC^f_\opt$ be the \emph{optimal} $k$ centers for $(\cC, \cW^f)$, and $\cC'_\opt$ optimal for $(\cC, \cW')$. We have $\phi_{(\cC, \cW^f)}(\cC'_k) \leq (1+2\epsilon) \phi_{(\cC, \cW')}(\cC'_k)$ for the reason that the contribution of each point grows by at most $(1+2\epsilon)$ due to weight approximation. Using the same analysis, $\phi_{(\cC, \cW')}(\cC^f_\opt) \leq (1+2\epsilon) \phi_{\cW^f}^*$. Combining the two inequalities, we have 
\begin{align}
\begin{aligned}
\label{mainthm:eq1}
\phi_{(\cC, \cW^f)}(\cC'_k) &\leq (1+2\epsilon)^2 \phi_{(\cC, \cW')}(\cC'_k) \leq (1+2\epsilon)^2\gamma \phi_{\cW'}^* \\ 
&\leq (1+2\epsilon)^2\gamma \phi_{(\cC, \cW')}(\cC^f_\opt)  &\mbox{[by optimality of $\phi_{\cW'}^*$]}\\
&\leq (1+2\epsilon)^3\gamma \phi_{\cW^f}^* \leq 322\gamma(1+2\eps)^4 \phi_\opt &\mbox{[using Lemma \ref{lem:const_appx_adaptive}]}
\end{aligned}
\end{align}

Let $\phi_{S}(\cC'_k) = \sum_{p \in S}\min_{c \in \cC'_k}\norm{p - c}^2$. For every point $p \in S$, to bound its cost $\min_{c \in \cC'_k}\|p - c\|^2$, we use multiple relaxed triangle inequalities for every center $c_i \in \cC$ , and take the weighted average of them using $\weight(p,c_i)$.

\begin{align*}
\phi_{S}(\cC'_k) &= \sum_{p \in S} \min_{c \in \cC'_k}\|p - c\|^2
\\
&= \sum_{p \in S} \sum_{c_i \in \cC} \weight(p, c_i) \min_{c \in \cC'_k}\|p - c\|^2 
&\mbox{[$\sum_{c_i \in \cC} \weight(p, c_i) = 1$]}
\\
&\leq \sum_{p \in S} \sum_{c_i \in \cC} \weight(p, c_i) \min_{c \in \cC'_k} 2(\|p - c_i\|^2 + \|c_i-c \|^2) 
& \mbox{[relaxed triangle inequality]}\\
& = 2\phi_\weight + 2\phi_{(\cC, \cW^f)}(\cC'_k)
&\mbox{[$\sum_{p \in S}\weight(p, c_i) = w_i^f$]}
\\
& \leq 2\phi_\weight + 2 \cdot 322\gamma(1+2\epsilon)^4 \phi_\opt
&\mbox{[inequality \eqref{mainthm:eq1}]}
\\
& \leq 2 \cdot 160 (1+\epsilon)\phi_\opt + 2 \cdot 322\gamma(1+2\epsilon)^4 \phi_\opt
& \mbox{[Lemma \ref{lem:bound-cost-by-alg}]}
\\
& = O(\gamma)\phi_\opt
\end{align*}

\end{proof}

\bibliographystyle{plain}
\bibliography{references}

\newpage
\appendix

\section{Relational Implementation of 3-means++}
\label{section:3means}
Recall that the 3-means++ algorithm picks a point $x$ 
to be the third center $c_3$ with probability
$P(x) = \frac{L(x)}{Y}$
where $L(x) = \min( \norm{x-c_1}_2^2, \norm{x - c_2}_2^2)$ and
$Y = \sum_{x \in J} L(x)$  is a normalizing constant.
Conceptually think of $P$ as being a 
`hard'' distribution to sample from. 

\textbf{Description of the Implementation:}
The implementation first constructs two identically-sized axis-parallel hypercubes/boxes $b_1$ and $b_2$ centered around $c_1$ and $c_2$ that are \textbf{as large as possible}
subject to the constraints that the side
lengths have to be non-negative integral powers of $2$,
and that $b_1$  and $b_2$ can not intersect. Such side lengths could be found since we may assume $c_1$ and $c_2$  have integer coordinates or they are sufficiently far away from each other that we can scale them and increase their distance.
Conceptually the implementation also considers a box $b_3$ that is the whole Euclidean space.  

\begin{figure}[h]
    \centering
\includegraphics[scale=0.3]{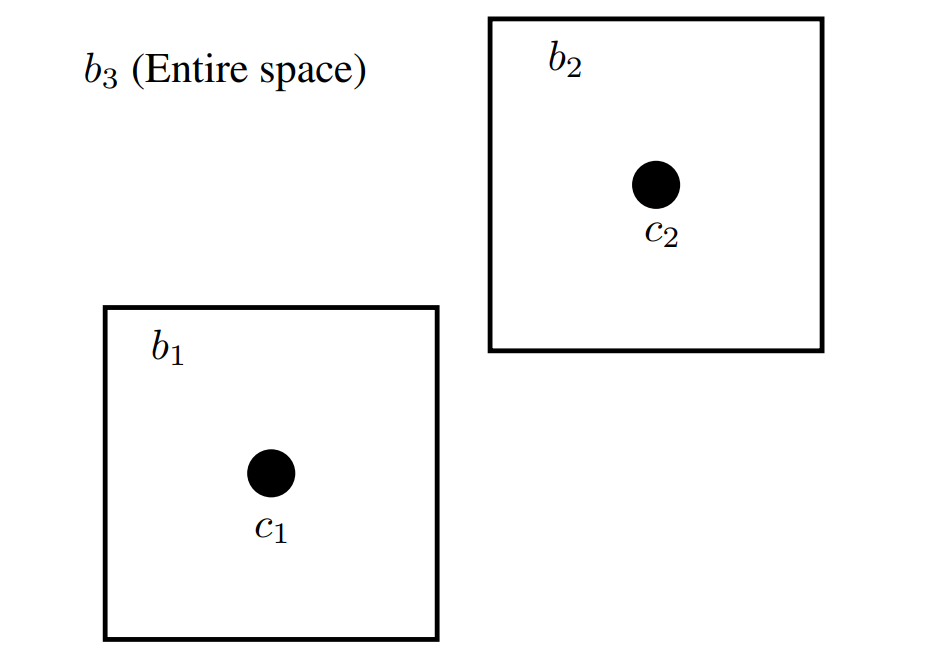}
    \caption{Boxes used for sampling the third center}
    \label{fig:third:center}
\end{figure}
To define our ``easy'' distribution $Q$, for each point $x$ define $R(x)$ to be
\begin{align*}
    R(x) = \begin{cases}
    \norm{x-c_1}_2^2    &    x\in b_1 \\
    \norm{x-c_2}_2^2   &   x \in b_2 \\
    \norm{x-c_1}_2^2    &   x \in b_3 \text{ and } x\notin b_1 \text{ and } x\notin b_2
    \end{cases}
\end{align*}
In the above definition, note that when $x\notin b_1 \text{ and } x\notin b_2$, the distance of $x$ to both centers are relatively similar; therefore, we can assign $x$ to either of the centers -- here we have assigned it to $c_1$. Then $Q(x)$ is defined to be $\frac{R(x)}{Z}$, where $Z = \sum_{x \in J} R(x)$ is normalizing constant.
The implementation then repeatedly samples
a point $x$ with probability  $Q(x)$. After sampling $x$, the implementation
can either (A) reject $x$, and then resample
or (B) accept $x$, which means setting the third center $c_3$ to be $x$. 
The probability that 
$x$ is accepted after it is sampled is $\frac{ L(x)}{R(x)}$, and thus the probability
that $x$ is rejected is $1-\frac{L(x)}{R(x)}$.

It is straightforward to see how to compute $b_1$ and $b_2$ (note that $b_1$ and $b_2$ can be computed without any relational operations), and how to compute $L(x)$ 
and $R(x)$ for a particular point $x$. 
Thus, the only non-straight-forward part is sampling a point $x$ with probability  $Q(x)$, which we explain now:
\begin{itemize}
    \item 
The implementation uses
a SumProd query  to compute the aggregate 
2-norm squared distance from $c_1$ constrained to
points in $b_3$ (all the points) and grouped 
by table $T_1$ using Lemma \ref{lem:box_computing}.
Let the resulting vector
be $C$. So $C_r$ is the aggregate 2-norm squared distance from $c_1$ of all rows in the design matrix that 
are extensions of row $r$ in $T_1$. 
\item
Then the implementation uses a SumProd query
to compute the aggregated 2-norm squared distance from $c_2$,
constrained to points in $b_2$, 
and grouped by $T_1$. Let the resulting
vector be $D$. Notice that an axis-parallel box constraint can be expressed as a collection of  axis-parallel hyperplane constraints, and for every axis-parallel constraint it is easy to remove the points not satisfying it from the join by filtering one of the input tables having that dimension/feature. Then the sum product query is the same as the sum product query in the previous step.
\item
Then the implementation uses a SumProd query
to compute the aggregated 2-norm squared distance  from $c_1$,
constrained to points in $b_2$,
and grouped by $T_1$
Let the resulting
vector be $E$. 
\item
Then pick a row $r$ of $T_1$ with probability proportional to $C_r - E_r + D_r$.
\item
The implementation then replaces $T_1$ by 
a table consisting only of the picked row $r$.
\item
The implementation then repeats this process
on table $T_2$, then table $T_3$ etc. 
\item
At the end $J$ will consist of one point/row $x$, where the probability that a particular point $x$ ends up as this final row is  $Q(x)$. 
To see this note that in the iteration performed for $T_i$, $C-E$ is the aggregate 2-norm squared distances to $c_1$ for all points not in $b_2$ grouped by $T_i$, and
$D$ is the aggregated squared distances of the
points in $b_2$ to $c_2$ grouped by $T_i$.
\end{itemize}

 We now claim that this implementation
 guarantees that $c_3=x$ with probability $P(x)$. We can see
 this using the standard rejection sampling
 calculation. At each iteration of sampling
 from $Q$, let $S(x)$ be the event that point $x$ is sampled and $A(x)$ be the event that $x$ is accepted. Then,
\begin{align*}
    \pr[S(x) \text{ and } A(x)] &= \pr[A(x)] \mid S(x)] \cdot \pr[S(x)]
    = \frac{L(x)}{R(x)} Q(x)
    = \frac{L(x)}{Z}
\end{align*}
Thus $x$ is accepted with probability proportional
to $L(x)$, as desired.

As the number of times that
the implementation has to sample from $Q$ is
geometrically distributed, the expected
number of times that it will have to sample
is the inverse of the probability of success, which is $\max_x \frac{R(x)}{L(x)}$. 
It is not too difficult to see (we prove
it formally in 
Lemma \ref{lemma:box:boundary}) that
$\max_x \frac{R(x)}{L(x)} = O(d)$.
It takes $3m$ SumProd queries to sample from $Q$. 
Therefore, the expected running time of
our implementation of 3-means++ is 
 $O(md \Psi(n, d, m))$.

\section{Pseudo-code}
\label{sect:pseudo-code}

In this section you may find the algorithms explained in Section \ref{sec:kmeans++} in pseudo-code format. 
\begin{algorithm}[H]
    \caption{Algorithm for creating axis-parallel hyperrectangles}
    \label{alg:finding:boxes}
    \begin{algorithmic}[1]
    \Procedure{Construct Boxes}{$\cC_{i-1}$}
        \State \textbf{Input:} Current centers $\cC_{i-1} = \{c_1, \dots, c_{i-1}\}$
        \State \textbf{Output:} $\cB_i$, a set of boxes and their centers
        \State $\cB_i \gets \emptyset$
        \State $\cG_i \gets \{(b^*_j,c_j) \:\vert\: b^*_j \text{ is a unit size hyper-cube around }c_j, j \in [i-1]\}$
        \Comment{We assume there is no intersection between the boxes in $G$ initially, up to scaling}
        \While{$|\cG_i| > 1$}
            \State Double all the boxes in $\cG_i$.
            \State $\cG'_i = \emptyset$ \Comment{Keeps the boxes created in this iteration of doubling}
            \While{$\exists (b_1,y_1) , (b_2,y_2) \in \cG_i$ that intersect with each other}
                \State $b \gets $ the smallest box in Euclidean space containing both $b_1$ and $b_2$.
                \State $\cG_i \gets (\cG_i \setminus \{(b_1,y_1),(b_2,y_2)\}) \cup \{(b,y_1)\}$
                \State $\cG'_i \gets (\cG'_i \cup \{(b,y_1)\}$
                \If{$(b_1,y_1)\notin \cG'_i$} \Comment{Check if box $b_1$ hasn't been merged with other boxes in the current round}
                    \State $b'_1 \gets$ halved $b_1$, add $(b'_1, y_1)$ to $\cB_i$
                \EndIf
                \If{$(b_2,y_2)\notin \cG'_i$} \Comment{Check if box $b_2$ hasn't been merged with other boxes in the current round}
                    \State $b'_2 \gets$ halved $b_2$, add  $(b'_2,y_2)$ to $\cB_i$
                \EndIf
            \EndWhile
        \EndWhile
        \State There is only one box and its representative remaining in $\cG_i$, replace this box with the whole Euclidean space.
        \State $\cB_i \gets \cB_i \cup \cG_i$.
        \State Return $\cB_i$.
    \EndProcedure
    \end{algorithmic}
\end{algorithm}

\section{Omitted Proofs}
\label{section:ommited:proofs}
\subsection*{NP-Hardness of Approximating Cluster Size}
\begin{proofof}[Theorem \ref{thm:hardcount}]

We've proved the \#P-hardness in the main body. Here we prove the second part of Theorem \ref{thm:hardcount} that given an acyclic database and a set of centers ${c_1,\dots,c_k}$, it is NP-Hard to approximate the number of points assigned to each center when $k\geq 3$. We prove it by reduction from Subset Sum. In Subset Sum problem, the input is a set of integers $A = {w_1,\dots,w_m}$ and an integer $L$, the output is true if there is a subset of $A$ such that its summation is $L$. We create the following acyclic schema. There are $m$ tables. Each table $T_i$ has a single unique column $x_i$ with two rows ${w_i,0}$. Then the join of the tables has $2^m$ rows, and it is a cross product of the rows in different tables in which each row represents one subset of $A$. 

Then consider the following three centers: $c_1 = (\frac{L-1}{m},\frac{L-1}{m},\dots,\frac{L-1}{m})$, $c_2 = (\frac{L}{m},\dots,\frac{L}{m})$, and $c_1 = (\frac{L+1}{m},\frac{L+1}{m},\dots,\frac{L+1}{m})$. The Voronoi diagram that separates the points assigned to each of these centers consists of two parallel hyperplanes: $\sum_i x_i = L-1/2$ and $\sum_i x_i = L+1/2$ where the points between the two hyperplanes are the points assigned to $c_2$. Since all the points in the design matrix have integer coordinates, the only points that are between these two hyperplanes are those points for which $\sum_i x_i = L$. Therefore, the approximation for the number of points assigned to $c_2$ is non-zero if and only if the answer to Subset Sum is True.
\end{proofof}

\begin{algorithm}[H]
    \caption{Algorithm for sampling the next center}
    \label{alg:sampling1}
    \begin{algorithmic}[1]
    \Procedure{KMeans++Sample}{$\cC_{i-1}, T_1, \dots, T_m$}
    \State Let $p(b)$ be the box that is the parent of $b$ in the tree structure of all boxes in $\cB_i$.
    \State $c_i \gets \emptyset$
    \State $\cB_i \gets \textsc{Construct Boxes}(\cC_{i-1})$
    \State Let $(b_0,y_0)$ be the tuple where $b_0$ is the entire Euclidean space in $\cB_i$.
    \While{$c_i = \emptyset$}
    \For{$1 \leq \ell \leq m$} \Comment{Sample one row from each table.}
        \State Let $H$ be a vector having an entry $H_r$ for each $r \in T_\ell$.
        \State $J' \gets r_1 \Join \ldots \Join r_{\ell-1} \Join J$. \Comment{Focus on only the rows in $J$ that uses all previously sampled rows from $T_1, \ldots, T_{\ell-1}$ in the concatenation.}
        \State $\forall r \in T_\ell$ evaluate $H_r \gets \sum_{x \in r \Join J' \cap b_0} \norm{x-y_0}_2^2$ 
        \For{$(b,y) \in \cB_i \setminus \{(b_0, y_0)\}$}
            \State Let $(b',y') \in \cB_i$ be the tuple where $b'=p(b)$.
            \State $\forall r \in T_\ell$ use SumProd query to evaluate two values: $\sum_{x \in r \Join J' \cap b}\norm{x-y}_2^2$ and $\sum_{x \in r \Join J' \cap b}\norm{x-y'}_2^2$.
            \State $H_r \gets H_r - \sum_{x \in r \Join J' \cap b} \norm{x-y'}_2^2 + \sum_{x \in r \Join J' \cap b} \norm{x-y}_2^2$ 
        \EndFor
        \State Sample a row $r_\ell \in T_\ell$ with probability proportional to $H_r$. 
    \EndFor
        \State $x \gets r_1 \Join \dots \Join r_m$.
        \State Let $(b^*, y^*)$ be the tuple where $b^*$ is the smallest box in $\cB_i$ containing $x$.
        \State $c_i \gets x$ with probability $\frac{\min_{c \in \cC_{i-1}}\|x- c\|_2^2}{\norm{x - y^*}_2^2}$. \Comment{Rejection sampling.}
    \EndWhile
    \State \textbf{return} $c_i$.
    \EndProcedure
    \end{algorithmic}
\end{algorithm}

\section{Uniform Sampling From a Hypersphere}
\label{sect:faqai}

In order to uniformly sample a point from inside a ball, it is enough to show how we can count the number of points located inside a ball grouped by a table $T_i$. Because, if we can count the number of points grouped by input tables, then we can use similar technique to the one used in Section \ref{sec:kmeans++} to sample. Unfortunately, as we discussed in Section \ref{subsec:warmup}, it is $\#P$-Hard to count the number of points inside a ball; however, it is possible to obtain a $1\pm \delta$ approximation of the number of points \cite{abokhamis2020approximate}. Bellow we briefly explain the algorithm in \cite{abokhamis2020approximate} for counting the number of points inside a hypersphere. 

Given a center $c$ and a radius $R$, the goal is approximating the number of tuples $x\in J$ for which $\sum_i (c^i-x^i)^2 \leq R$. Consider the set $S$ containing all the multisets of real numbers. We denote a multiset $A$ by a set of pairs of $(v,f_A(v))$ where $v$ is a real value and $f(v)$ is the frequency of $v$ in $A$. For example, $A=\{(2.3,10),(3.5,1)\}$ is a multiset that has $10$ members with value $2.3$ and $1$ member with value $3.5$. Then, let $\oplus$ be the summation operator meaning $C = A \oplus B$ if and only if for all $x \in R$, $f_C(x) = f_A(x) + f_B(x)$, and let $\otimes$ be the convolution operator such that $C = A \otimes B$ if and only if $f_C(x) = \sum_{i\in \mathbb{R}} f_A(i) + f_B(x-i)$. Then the claim is $(S,\oplus,\otimes)$ is a commutative semiring and the following SumProd query returns a multiset that has all the squared distances of the points in $J$ from $C$:
\begin{align*}
    \bigoplus_{x\in J}\bigotimes_{i}\{((x^i-c^i)^2, 1)\}
\end{align*}
Using the result of the multiset, it is possible to count exactly the number of tuples $x\in J$ for which $\norm{x-c}_2^2\leq R^2$. However, the size of the result is as large as $\Omega(|J|)$.

In order to make the size of the partial results and time complexity of $\oplus$ and $\otimes$ operators polynomial, the algorithm uses $(1 + \delta)$ geometric bucketing. The algorithm returns an array where in $j$-th entry it has the smallest value $r$ for which there are $(1+\delta)^j$ tuples $x\in J$ satisfying $\norm{x-c}_2^2 \leq r^2$. 

The query can also be executed grouped by one of the input tables. Therefore, using this polynomial approximation scheme, we can calculate conditioned marginalized probability distribution with multiplicative $(1 \pm \delta)$. Therefore, using $m$ queries, it is possible to sample a tuple from a ball with probability distribution $\frac{1}{n}(1\pm m \delta)$ where $n$ is the number of points inside the ball. In order to get a sample with probability $\frac{1}{n}(1\pm \epsilon)$, all we need is to set $\delta = \epsilon/m$; hence, on \cite{abokhamis2020approximate}, the time complexity for sampling each tuple will be $O\Big(\frac{m^9\log^4(n)}{\epsilon^2}\Psi(n,d,m)\Big)$

\section{Hardness of Lloyd's Algorithm}
\label{sect:Lloyds}
After choosing  $k$ initial centers, a type of local search algorithm,
called Lloyd's algorithm,  is commonly used to iteratively find better centers. After associating each point with its closest center, and Lloyd's algorithm updates the position of each center to the center of mass of its associated points. Meaning, if $X_c$ is the set of points assigned to $c$, its location is updated to $\frac{\sum_{x\in X_c} x}{|X_c|}$. While this can be done easily when the data is given explicitly,  we show in the following theorem that finding the center of mass for the points assigned to a center is $\#$P-hard when the data is relational,
even in the special case of an acyclic join and two centers. 

\begin{theorem}
Given an acyclic join, and two centers, it is $\#$P-hard to compute the center of mass for the points assigned to each center.
\end{theorem}
\begin{proof}
We prove by a reduction from  a decision version
of the counting knapsack problem. 
The input to the counting knapsack problem
consists of a the set  $W= \{w_1,\dots, w_n\}$
of positive integer weights,  a knapsack size $L$, and a count $D$. The problem is to
determine whether there are at least $D$ subsets of $W$ with
aggregate weight at most $L$. The points in our
instance of $k$-means will be given relationally. 
We construct a join query with $n+1$ columns/attributes, and $n$ tables. 
All the tables have one column in common and one distinct column. The $i$-th table has $2$ columns $(d_i, d_{n+1})$ and three rows $\{(w_i,-1), (0,-1), (0, D)\}$. Note that the join has $2^n$ rows with $-1$ in dimension $n+1$, and one row with values $(0,0,\dots,0,D)$. The rows with $-1$ in dimension $d+1$ have all the subsets of $\{w_1,\dots, w_n\}$ in their first $n$ dimensions. Let the two centers for $k$-means problem be any two centers $c_1$ and $c_2$ such that a point $x$ is closer to $c_1$ if it satisfies $\sum_{d=1}^n x_d <  L$ and  closer to $c_2$ if it satisfies $\sum_{d=1}^n x_d >  L$. Note that the row $(0,0,\dots,0,D)$ is closer to $c_1$. Therefore, the value of dimension $n+1$ of the center of mass for the tuples that are closer to $c_1$ is $Y= (D-C)/C$ where $C$ is the actual number of subsets of $W$ with aggregate weight at most $L$. If $Y$ is negative, then the number of solutions to the counting knapsack instance is at least $D$.
\end{proof}

\section{Background Information About Database Concepts} 
\label{sect:dbbackground}

Given a tuple $x$, define $\Pi_{F}(x)$ to be projection of $x$ onto the set of features $F$ meaning $\Pi_{F}(x)$ is a tuple formed by keeping the entries in $x$ that are corresponding to the feature in $F$. For example let $T$ be a table with columns $(A,B,C)$ and let $x = (1,2,3)$ be a tuple of $T$, then $\Pi_{\{A,C\}}(x) = (1,3)$. 

\begin{definition}[Join]
Let $T_1,\dots, T_m$ be a set of tables with corresponding sets of columns/features $F_1,\dots,F_m$ we define the join of them $J=T_1 \Join \dots \Join T_m$ as a table such that the set of columns of $J$ is $\bigcup_i F_i$, and $x\in J$ if and only if $\Pi_{F_i}(x) \in T_i$.
\end{definition}

Note that the above definition of join is consistent with the definition written in Section \ref{section:intro} but offers more intuition about what the join operation means geometrically.

\begin{definition}[Join Hypergraph]
Given a join $J=T_1 \Join \dots \Join T_m$, the hypergraph associated with the join is $H=(V,E)$ where $V$ is the set of vertices and for every column $a_i$ in $J$ there is a vertex $v_i$ in $V$, and for every table $T_i$ there is a hyper-edge $e_i$ in $E$ that has the vertices associated with the columns of $T_i$.
\end{definition}

\begin{theorem}[AGM Bound \cite{atserias2008size}]
\label{appendix:agm}
Given a join $J=T_1 \Join \dots \Join T_m$ with $d$ columns and its associated hypergraph $H=(V,E)$, and let $C$ be a subset of $\col(J)$, let $X = (x_1, \dots, x_m)$ be any feasible solution to the following Linear Programming:
\begin{alignat*}{3}
 & \text{minimize} & \sum_{j=1}^{m} \log(|T_j|)x_{j}& \\
 & \text{subject to} \quad& \sum_{\mathclap{{j:v \in e_{j}}}}x_{j}& \geq 1, & v \in C\\
                 && 0 \leq x_{j}& \leq 1,\quad & j &=1 ,..., t
\end{alignat*}
Then $\prod_i |T_i|^{x_i}$ is an upper bound for the cardinality of $\Pi_C(J)$, this upperbound is tight if $X$ is the optimal answer.
\end{theorem}

We give another definition of \emph{acyclicity} which is consistent with the definition in the main body.
\begin{definition}[Acyclic Join]
We call a join query (or a relational database schema) \textbf{acyclic} if one can repeatedly apply one of the two operations and convert the set of tables to an empty set:
\begin{enumerate}
    \item Remove a column that is only in one table.
    \item Remove a table for which its columns are fully contained in another table.
\end{enumerate}
\end{definition}

\begin{definition}[Hypertree Decomposition]
Let $H=(V,E)$ be a hypergraph and $T=(V',E')$ be a tree with a subset of $V$ associated to each vertex in $v' \in V'$ called \textbf{bag} of $v'$ and show it by $b(v') \subseteq V$. $T$ is called a \textbf{hypertree decomposition} of $H$ if the following holds:
\begin{enumerate}
    \item For each hyperedge $e \in E$ there exists $v' \in V'$ such that $e \subseteq b(v')$
    \item For each vertex $v \in E$ the set of vertices in $V'$ that have $v$ in their bag are all connected in $T$.
\end{enumerate}
\end{definition}

\begin{definition}
Let $H=(V,E)$ be a join hypergraph and $T=(V',E')$ be its hypertree decomposition. For each $v' \in V'$, let $X^{v'} = (x_1^{v'},x_2^{v'},\dots, x_m^{v'})$ be the optimal solution to the following linear program: $\texttt{min}  \sum_{j=1}^{t} x_{j}$,  $\text{subject to }  \sum_{j:v_{i} \in e_{j} }x_{j} \geq 1, \forall v_i \in b(v')$ where $0 \leq x_{j} \leq 1$ for each $j \in [t]$.
Then the \textbf{width of $v'$} is $\sum_i x^{v'}_i$ denoted by $w(v')$ and the \textbf{fractional width of $T$} is $\max_{v' \in V'} w(v')$.
\end{definition}

\begin{definition}[fhtw]
Given a join hypergraph $H=(V,E)$, the \textbf{fractional hypertree width of $H$}, denoted by fhtw, is the minimum fractional width of its hypertree decomposition. Here the minimum is taken over all possible hypertree decompositions.
\end{definition}

\begin{observation}
The fractional hypertree width of an \allowbreak acyclic join is $1$, and each bag in its hypertree decomposition is a subset of the columns in some input table.
\end{observation}


\begin{theorem}[Inside-out \cite{khamis2016joins}]
\label{thm:insideout}
There exists an algorithm to evaluate a SumProd query in time $O(T md^2 n^{\fhtw} \log(n))$ where $\fhtw$ is the fractional hypertree width of the query and $T$ is the time needed to evaluate $\oplus$ and $\otimes$ for two operands. The same algorithm with the same time complexity can be used to evaluate SumProd queries grouped by one of the input tables.
\end{theorem}

\begin{theorem}
\label{thm:sumsum:query}
Let $Q_f$ be a function from domain of column $f$ in $J$ to $\mathbb{R}$, and $G$ be a vector that has a row for each tuple $r\in T_i$. Then the query 
\begin{align*}
    \sum_{X\in J}\sum_{f} Q_f(x_f)
\end{align*}
can be converted to a SumProd and the query returning $G$ with definition
\begin{align*}
   G_r = \sum_{X\in Y_i \Join J} \sum_{f} F_i(x_f)
\end{align*}
can be converted to a SumProd query grouped by $T_i$.
\end{theorem}
\begin{proof}
Let $S = \{(a,b) \:\vert\: a \in \mathbb{R}, b \in \mathbb{I}\}$, and for any two pairs of $(a,b),(c,d) \in S$ we define:
\begin{align*}
    (a,b) \oplus (c,d) = (a+c,b+d)
\end{align*}
and
\begin{align*}
    (a,b) \otimes (c,d) = (ad+cb, bd).
\end{align*}

Then the theorem can be proven by using the following two claims:
\begin{enumerate}
    \item $(S,\oplus,\otimes)$ forms a commutative semiring with identity zero $I_0 = (0,0)$ and identity one $I_1 = (0,1)$.
    \item The query $\oplus_{X \in J} \otimes_{f} (Q_f(x_f),1)$ is a SumProd FAQ where the first entry of the result is $\sum_{X\in J}\sum_{f} Q_f(x_f)$ and the second entry is the number of rows in $J$.
\end{enumerate}

proof of the first claim: Since arithmetic summation is commutative and associative, it is easy to see $\oplus$ is also commutative and associative. Furthermore, based on the definition of $\oplus$ we have $(a,b) \oplus I_0 = (a+0,b+0) = (a,b)$.

The operator $\otimes$ is also commutative since arithmetic multiplication is commutative, the associativity of $\otimes$ can be proved by
\begin{align*}
    (a_1,b_1) \otimes ((a_2,b_2) \otimes (a_3,b_3)) &= (a_1,b_1) \otimes (a_2 b_3 + a_3 b_2, b_2 b_3)
    \\ &= (a_1 b_2 b_3 + b_1 a_2 b_3 + b_1 b_2 a_3, b_1 b_2 b_3)
    \\ &= (a_1 b_2 + b_1 a_2, b_1 b_2) \otimes (a_3,b_3)
    \\&= ((a_1,b_1) \otimes (a_2,b_2)) \otimes (a_3,b_3)
\end{align*}

Also note that based on the definition of $\otimes$, $(a,b) \otimes I_0 = I_0$ and $(a,b) \otimes I_1 = (a,b)$. The only remaining property that we need to prove is the distribution of $\otimes$ over $\oplus$:
\begin{align*}
    (a,b) \otimes ((c_1,d_1) \oplus (c_2,d_2)) &= (a,b) \otimes (c_1+c_2,d_1 + d_2)
    \\ &= (a,b) \otimes (c_1 + c_2, d_1 + d_2) 
    \\ &= (c_1 b + c_2 b + a d_1 + a d_2, bd_1 + bd_2)
    \\ &= (c_1 b + a d_1, bd_1) \oplus (c_2 b + a d_2, bd_2)
    \\ &= ((a,b) \otimes (c_1,d_1)) \oplus ((a,b) \otimes (c_2,d_2)) 
\end{align*}

Now we can prove the second claim: To prove the second claim, since we have already shown the semiring properties of $(S, \oplus, \otimes)$ we only need to show what is the result of $\oplus_{X \in J} \otimes_{f} (Q_f(x_f),1)$. We have $\otimes_{f} (Q_i(x_f),1) = (\sum_f Q_i(x_f), 1)$, therefore
\begin{align*}
    \oplus_{X \in J} \otimes_{f} (Q_i(x_f),1) = \oplus_{X \in J} (\sum_f Q_f(x_f), 1) = (\sum_{X\in J}\sum_{f} Q_f(x_f), \sum_{X \in J} 1)
\end{align*}
where the first entry is the result of the SumSum query and the second entry is the number of rows in $J$. 
\end{proof}

\end{document}